\documentclass[aps,pra,superscriptaddress,twocolumn,longbibliography]{revtex4-1}
\usepackage{graphicx, color, graphpap}
\usepackage{enumitem}
\usepackage{amssymb}
\usepackage{amsthm}
\usepackage{amsmath}
\usepackage{dsfont}
\usepackage{mathtools}
\usepackage{soul}

\usepackage[dvipsnames]{xcolor}
\usepackage{multirow}
\usepackage[colorlinks=true,citecolor=blue,linkcolor=magenta]{hyperref}
\usepackage[T1]{fontenc}

\usepackage{thmtools,thm-restate}
\usepackage{verbatim}
\usepackage{titlesec}

\usepackage{graphicx}

\usepackage{bbm}
\usepackage{titlesec}

\usepackage[ruled,vlined]{algorithm2e}

\newtheorem{theorem}{Theorem}
\newtheorem{lemma}{Lemma}
\newtheorem{definition}{Definition}
\newtheorem{corollary}{Corollary}

\newcommand{\bra}[1]{\langle#1|}

\newcommand{\ket}[1]{|#1\rangle}

\usepackage{etoolbox}
\makeatletter
\patchcmd{\tableofcontents}
  {\relax}
  {\@starttoc{toc}}
  {}{}
\makeatother

\newcommand{\resetAppendixCounters}[1]{%
  \setcounter{lemma}{0}
  \renewcommand{\thelemma}{#1.\arabic{lemma}}
  \setcounter{proposition}{0}
  \renewcommand{\theproposition}{#1.\arabic{proposition}}
  \setcounter{theorem}{0}
  \renewcommand{\thetheorem}{#1.\arabic{theorem}}
  \setcounter{corollary}{0}
  \renewcommand{\thecorollary}{#1.\arabic{corollary}}
  \setcounter{definition}{0}
  \renewcommand{\thedefinition}{#1.\arabic{definition}}
  \setcounter{subsection}{0}

  \renewcommand{\theequation}{#1.\arabic{equation}}
  \renewcommand{\thefigure}{#1.\arabic{figure}}
  \renewcommand{\theHfigure}{#1.\arabic{figure}} 
  \setcounter{figure}{0}
  \setcounter{equation}{0}
}

\begin{document}

\title{Tight Generalization Bound for Supervised Quantum Machine Learning}

\author{Xin Wang}
\affiliation{
	Department of Automation, Tsinghua University, Beijing, 100084, P. R. China
}

\author{Rebing Wu}
\email{rbwu@tsinghua.edu.cn}
\affiliation{
	Department of Automation, Tsinghua University, Beijing, 100084, P. R. China
}

\begin{abstract}
  We derive a tight generalization bound for quantum machine learning that is applicable to a wide range of supervised tasks, data, and models. Our bound is both efficiently computable and free of big-O notation.  Furthermore, we point out that previous bounds relying on big-O notation may provide misleading suggestions regarding the generalization error. Our generalization bound demonstrates that for quantum machine learning models of arbitrary size and depth, the sample size is the most dominant factor governing the generalization error. Additionally, the spectral norm of the measurement observable, the bound and Lipschitz constant of the selected risk function also influence the generalization upper bound. However, the number of quantum gates, the number of qubits, data encoding methods, and hyperparameters chosen during the learning process such as batch size, epochs, learning rate, and optimizer do not significantly impact the generalization capability of quantum machine learning. We experimentally demonstrate the tightness of our generalization bound across classification and regression tasks. Furthermore, we show that our tight generalization upper bound holds even when labels are completely randomized. We thus bring clarity to the fundamental question of generalization in quantum machine learning. 
\end{abstract}

\maketitle

\noindent\textbf{\textit{Introduction}.--}Quantum machine learning (QML)~\cite{biamonte2017quantum} has shown significant promise in harnessing quantum properties to achieve advantages in supervised machine learning tasks~\cite{riste2017demonstration,huang2022quantum,cerezo2022challenges,liu2021rigorous,huang2021information}. The contemporary QML paradigm primarily combines parameterized quantum circuits with task-specific observables to form QML models~\cite{benedetti2019parameterized}, which are then trained by a classical optimizer to process both classical and quantum data~\cite{beer2020training,abbas2021power}. QML has demonstrated remarkable potential in addressing challenges across diverse domains, including quantum physics~\cite{monaco2023quantum,feng2025uncovering}, financial analysis~\cite{thakkar2024improved,doosti2024brief}, image classification~\cite{senokosov2024quantum,wei2023quantum}, and beyond. While practical applications  continue to emerge, theoretical analysis of QML's fundamental properties provides even greater value by deepening our understanding and illuminating the path forward for future applications in this interdisciplinary field~\cite{schuld2022quantum}.

Supervised Machine learning aims to achieve good performance on both training data and unseen data, with generalization capability measuring the difference between a model's predictive performance on training and unseen data~\cite{mohri2018foundations}. Small training error combined with small generalization error ensures good predictive performance. Since generalization error cannot be directly measured, theoretical generalization bounds become indispensable, as they link the error to quantifiable factors like sample size and model complexity, offering critical guidance for practice. For QML tasks, existing generalization bounds face several issues: (1) the bounds contain big-O notation, omitting potentially large constant factors, resulting in practically useless upper bounds that may provide misleading guidance regarding the number of parameterized quantum gates~\cite{caro2022generalization}, encoding methods~\cite{caro2021encodingdependent}, and optimization approaches~\cite{yang2025stability};  (2) the bounds involve expectations~\cite{caro2024information}, mutual information~\cite{banchi2021generalization}, quantum Fisher information~\cite{khanal2025data}, and other quantities that are difficult to compute or estimate~\cite{hur2024understanding}, and hence their practical utility is limited; and (3) the bounds depend on specific forms of risk functions~\cite{caro2022generalization,banchi2021generalization}, failing to provide unified analysis for different risk functions in QML.

In this work, we derive a tight uniform generalization bound for QML that is free of big-O notation, provides directly actionable information, and applies to general risk functions. This generalization bound reveals a surprising fact: for QML models, increasing parameterized quantum gates, increasing data dimensionality, changing classical data encoding methods, or modifying model optimization strategies, the generalization capability cannot be significantly altered. The generalization capability depends solely on the training sample size, as shown in Fig.~\ref{fig:illustration}. Our generalization bound explains the good generalization characteristics observed in existing QML works~\cite{schuld2020circuit,hubregtsen2022training,henderson2020quanvolutional,xu2024quantum,liu2021rigorous,caro2022generalization,wang2024supervised}.

\begin{figure*}[htpb]
  \centering
  \includegraphics[width=0.95\textwidth]{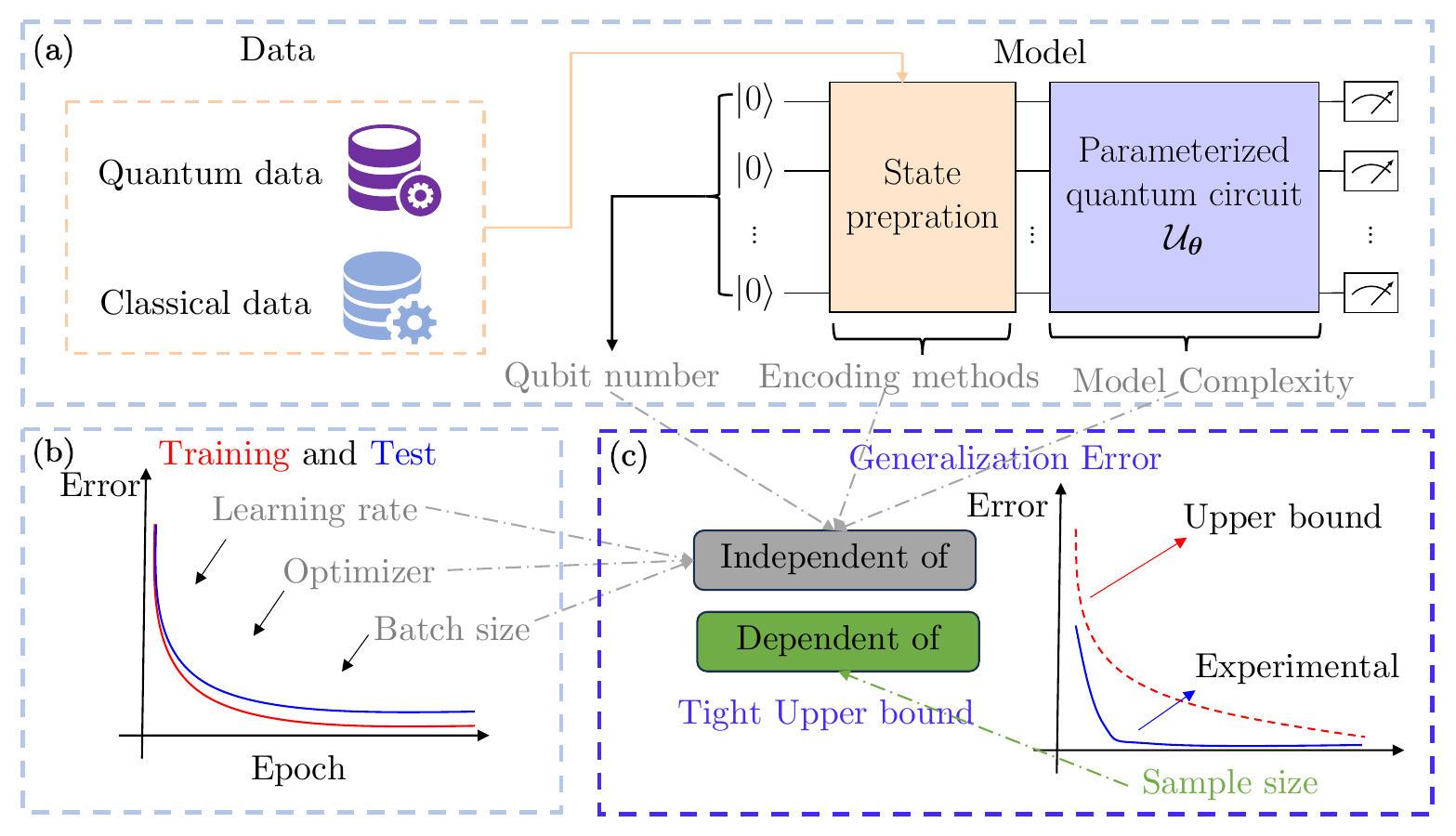}
  \caption{(a) Quantum machine learning workflow: Quantum data (quantum states) are prepared using  quantum circuits, or classical data are encoded into quantum states through some encoding scheme. The input is then processed through parameterized quantum circuits for learning, and the output are obtained through measurements. The model characteristics include number of qubits, encoding methods, and model complexity. (b) Training and test errors decrease as training epochs increase. During the training process of QML models, the choice of learning rate, optimizer, and batch size may all influence the training dynamics. (c) Our derived theoretical generalization error upper bound is tight and depends only on sample size, independent of the number of qubits, encoding methods, model complexity, learning rate, optimizer, and batch size.}
  \label{fig:illustration}
\end{figure*}

\noindent\textbf{\textit{Supervised QML}}.--- In the context of supervised QML, we consider a training dataset of $M$ samples $S = \{(\rho^{(m)},y^{(m)})\}_{m=1}^M$, where $\rho^{(m)} \in \mathbb{H}$ represents a quantum state density matrix in Hilbert space $\mathbb{H}$, and $y^{(m)} \in \mathcal{Y}$ represents its corresponding label from the label space $\mathcal{Y}$. Each sample $(\rho^{(m)},y^{(m)})$ is drawn from an unknown joint distribution $\mathcal{D}$ over the product space of Hilbert space $\mathbb{H}$ and label space $\mathcal{Y}$. The goal of QML is to learn a hypothesis $h_S \in \mathcal{H}$ from the training set $S$ that can generalize well to unseen data sampled from the same unknown distribution $\mathcal{D}$, where $\mathcal{H}$ is hypothesis space.

In this work, we consider quantum states that can be either direct quantum data, such as ground states of Hamiltonians, or quantum states obtained from classical data through angle encoding, amplitude encoding, or other encoding methods. We consider a broader class of QML models that can be any parameterized quantum circuit, not just parameterized quantum circuits. A QML model $\mathcal{U}_{\boldsymbol{\theta}}(\cdot)$ with parameters $\boldsymbol{\theta}$, when given an input quantum state $\rho$, produces an output $h(\rho,\boldsymbol{\theta}) = \operatorname{Tr}\left[O \mathcal{U}_{\boldsymbol{\theta}}(\rho)\right]$ with respect to an observable $O$. The hypothesis learned from the training set is $h_S(\rho) = h(\rho,\boldsymbol{\theta}^{*})$, where $\boldsymbol{\theta}^{*}$ are the parameters determined during the optimization process.

In machine learning, we evaluate the predictive performance of a hypothesis $h_S$ on a quantum state $\rho$ using a per-sample risk function $r(y',y) = r(h_S(\rho),y)$, where $y' = h_S(\rho) \in \mathcal{Y}$ is the predicted label by $h_S$ and $y$ is the true label. We define the empirical error of the hypothesis $h_S$ learned from training set $S$ on a dataset $S' = \{(\rho'^{(m)},y'^{(m)})\}_{m=1}^M$ as:
$$
\begin{aligned}
\widehat{R}_{S'}(h_S) = \frac{1}{M} \sum_{m=1}^M r(h_S(\rho'^{(m)}),y'^{(m)}).
\end{aligned}
$$
When $S' = S$, the empirical error $\widehat{R}_S(h_S)$ is the training error. When the dataset $S'$ contains completely unseen data is drawn from the same distribution as the training set $S$, we call $\widehat{R}_{S'}(h_S)$ the test error and $S'$ the test set. It is worth noting that to reduce the training error, we minimize a selected loss function $\mathcal{L}(\boldsymbol{\theta};S)$ by updating model parameters $\boldsymbol{\theta}$ on the training set $S$. Here, the loss function differs from the training error: the loss function is primarily used for training, while the training error and per-sample risk function $r(\cdot,\cdot)$ are mainly used to assess the training results. The prediction error of the hypothesis $h_S$ learned from training set $S$ on the distribution $\mathcal{D}$ is defined as:
$$
\begin{aligned}
R(h_S) = \underset{(\rho,y) \sim \mathcal{D}}{\mathbb{E}}[r(h_S(\rho),y)].
\end{aligned}
$$
 The fundamental goal of machine learning is to ensure that models with good training performance also perform well on unseen data, which is known as generalization capability. To quantify a model's generalization capability, we define the generalization error $\operatorname{gen}(h_S)$ as:
$$
\begin{aligned}
\operatorname{gen}(h_S) = R(h_S) - \widehat{R}_S(h_S).
\end{aligned}
$$
When the risk function $r(\cdot,\cdot)$ has a range of $[0,C]$, the generalization error has a range of $[-C,C]$. Although the generalization error can be negative, indicating that the prediction error is smaller than the training error, meaning the model performs better on unseen data than on training data, we are more concerned with cases where the prediction error exceeds the training error. Therefore, in theoretical analysis, we primarily focus on the upper bound of the generalization error $\operatorname{gen}(h_S)$. Since the prediction error cannot be measured directly as the data distribution D is usually unknown, we typically estimate it using the test error $\widehat{R}_{S'}(h_S)$ to further evaluate the generalization error.

\noindent\textbf{\textit{Pauli Basis Representation in Quantum Machine Learning}.--} In this section, we convert quantum states, parameterized quantum circuits, and quantum measurements into representations in the Pauli basis to better derive generalization bounds.

Consider the density matrix $\rho$ of an $N$-qubit quantum state, which can be decomposed in the Pauli basis as:
\begin{equation}
  \label{eq:pauli}
\begin{aligned}
\rho = \frac{1}{2^{N}} \left( \sum_{P_i \in \{I,Z,Y,X\}^{\otimes N}}^{} \alpha_i P_i \right) ,
\end{aligned}
\end{equation}
where $\alpha_i = \operatorname{Tr}\left[\rho P_i\right]$ represents the coefficient of the Pauli basis $P_i$. Since $P_i$ is Hermitian, $\alpha_i$ is necessarily real, and the summation encompasses $4^{N}$ terms. Consequently, any quantum state $\rho$ can be uniquely represented as a coefficient vector $\boldsymbol{\alpha} = \begin{bmatrix} 
    \alpha_1 & \cdots & \alpha_{4^N} 
\end{bmatrix}^{\top}$ determined by these $4^{N}$ Pauli coefficients. This representation is unique because a quantum state is completely characterized by its expectation values $\operatorname{Tr}[\rho P_i]$ over a complete set of observables, such as the full Pauli basis. The squared $\ell^{2}$-norm of the Pauli coefficient vector $\boldsymbol{\alpha}$ is proportional to the purity of the quantum state. Specifically, $\|\boldsymbol{\alpha}\|_{2}^{2} = 2^{N} \operatorname{Tr}[\rho^2]$~\cite{wangpredictive}. Since we consider input states that are pure states in quantum circuits, it follows that $\|\boldsymbol{\alpha}\|_{2}^2 = 2^{N}$.

Consider a quantum state $\rho$ with its corresponding Pauli coefficient vector $\boldsymbol{\alpha}$. After passing through a parameterized quantum circuit, the resulting quantum state $\mathcal{U}_{\boldsymbol{\theta}}(\rho) = U(\boldsymbol{\theta}) \rho U(\boldsymbol{\theta})^{\dagger}$ has a corresponding Pauli coefficient vector $\boldsymbol{\beta}$. Due to the linearity of quantum circuits with respect to quantum states $\rho$, we have $\mathcal{U}_{\boldsymbol{\theta}}(c_{1} \rho_{1} + c_{2} \rho_2) = c_{1} \mathcal{U}_{\theta}(\rho_1) + c_{2} \mathcal{U}_{\theta}(\rho_2)$ and the corresponding Pauli coefficient transfer mapping $\mathcal{T}_{\boldsymbol{\theta}}(\cdot)$ satisfies $\mathcal{T}_{\boldsymbol{\theta}}(c_{1} \boldsymbol{\alpha}_1 + c_{2} \boldsymbol{\alpha}_2) = c_{1} \boldsymbol{\beta}_1 + c_{2} \boldsymbol{\beta}_2 = c_{1} \mathcal{T}_{\boldsymbol{\theta}}(\boldsymbol{\alpha}_1) + c_{2} \mathcal{T}_{\boldsymbol{\theta}}(\boldsymbol{\alpha}_2)$. Therefore, the mapping $\mathcal{T}_{\boldsymbol{\theta}}$ is linear and can be represented by a transfer matrix $T(\boldsymbol{\theta})$ in the Pauli basis. Here, the form of the transfer matrix $T$ reflects the architecture of the parameterized quantum circuit, while $\boldsymbol{\theta}$ reflects its parameters.

As shown in Theorem~C.1 of \cite{wangpredictive}, the Pauli transfer matrix of a quantum circuit is orthogonal. This property arises from the purity-preserving nature of unitary evolution: for a pure state $\rho$ and its evolved state $\mathcal{U}_{\boldsymbol{\theta}}(\rho)$ under a quantum circuit, their corresponding Pauli coefficient vectors $\boldsymbol{\alpha}$ and $\boldsymbol{\beta}$ are related by $\boldsymbol{\beta} = T(\boldsymbol{\theta}) \boldsymbol{\alpha}$. Since unitary evolution preserves purity, we have $\|\boldsymbol{\beta}\|_2^2 = \boldsymbol{\alpha}^\top T(\boldsymbol{\theta})^\top T(\boldsymbol{\theta}) \boldsymbol{\alpha} = \|\boldsymbol{\alpha}\|_2^2$. This equality holds for all $\boldsymbol{\alpha}$ if and only if $T(\boldsymbol{\theta})^\top T(\boldsymbol{\theta}) = I$. It is worth noting that while the transfer matrix of any quantum circuit is orthogonal, not all orthogonal matrices correspond to realizable quantum circuits.

\begin{figure*}[htpb]
  \centering
  \includegraphics[width=0.9\textwidth]{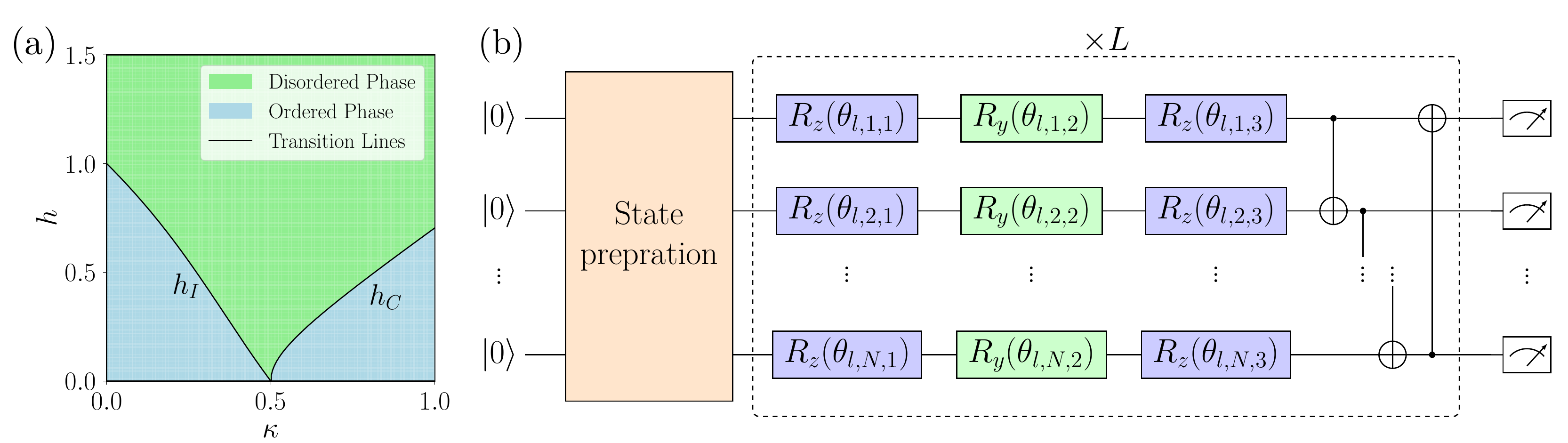}
  \caption{(a) Phase diagram of the axial next-nearest-neighbor Ising (ANNNI) model, illustrating the boundaries between ordered and disordered phases in the parameter space defined by  $\kappa$ and  $h$; (b) Parameterized quantum circuit architecture for  phase classification consisting of $L$ layers of rotation gates and controlled gates, where each qubit undergoes rotations $R_z(\theta_{1}) R_y(\theta_2) R_z(\theta_3)$ followed by ring-pattern CNOT gates creating entanglement.}
  \label{fig:experiments}
\end{figure*}

For the observable $O$, we adopt tensor products of Pauli operators commonly used in QML~\cite{wang2025limitations,li2022concentration}, such as $B_O \cdot Z \otimes I \otimes I$, $B_O \cdot X \otimes Y \otimes I$, or $B_O \cdot X \otimes Y \otimes Z$, where $X$, $Y$, $Z$ denote Pauli matrices and $B_O$ is the spectral norm of $O$. Since the observable operator is a Pauli string corresponding to the Pauli basis $P_i$ in Eq.~\eqref{eq:pauli}, we can represent it in the unnormalized Pauli basis (without the factor $1 / 2^{N}$). The Pauli coefficient vector corresponding to the observable is $\boldsymbol{m} = \begin{bmatrix} 
    0 & \cdots & 1 & \cdots & 0
\end{bmatrix}^{\top}$, which equals $1$ only at the $i$-th position and $0$ at all other positions. In the Pauli basis representation, the measurement result of quantum state $\rho$ with respect to the observable $O$, denoted as $\operatorname{Tr}[O \rho]$, can be elegantly expressed as the product of their coefficient vectors $B_{O} \cdot \boldsymbol{m}^{\top} \boldsymbol{\alpha}$, where $|\boldsymbol{m}^{\top} \boldsymbol{\alpha}| \leqslant 1$ due to the fact that the expectation value of any observable is bounded by its spectral norm. For more detailed proofs of this representation, please refer to SM.~\ref{asec:pauli_basis}.

In summary, the output of a QML model can be decomposed into three components: the input state's Pauli coefficient vector, the transfer matrix of the parameterized quantum circuit, and the Pauli basis representation of the observable. This output can be represented as $h(\rho,\boldsymbol{\theta}) = \operatorname{Tr}\left[O \mathcal{U}_{\boldsymbol{\theta}}(\rho)\right] = B_O \cdot \boldsymbol{m}^{\top} T(\boldsymbol{\theta}) \boldsymbol{\alpha} = B_O \cdot \boldsymbol{w}(T,\boldsymbol{\theta})^{\top} \boldsymbol{\alpha}$, where $\boldsymbol{w} = \boldsymbol{w}(T,\boldsymbol{\theta}) = T ^{\top}(\boldsymbol{\theta}) \boldsymbol{m}$ and $\|\boldsymbol{w}\|_{2} = 1$ due to the orthogonality of $T(\boldsymbol{\theta})$. Therefore, the hypothesis $h_S(\rho)$ can be represented as a function of the Pauli coefficient vector as $h_S(\boldsymbol{\alpha})$.

\noindent\textbf{\textit{Generalization Bound for QML}.--} In this section, we firstly introduce the basis concept of bound of generalization error, then derive the generalization bound for QML in different scenarios.

In machine learning theory, the generalization error can be bounded using the complexity of the hypothesis space, typically measured by Rademacher complexity~\cite{mohri2018foundations}. The empirical Rademacher complexity of a hypothesis space $\mathcal{H}$ with respect to a dataset sample $S = \{(x^{(m)},y^{(m)})\}_{m=1}^M$ is defined as $$\widehat{\mathfrak{R}}_{S}(\mathcal{H})=\underset{\boldsymbol{\sigma}}{\mathbb{E}}\left[\sup _{h \in \mathcal{H}} \frac{1}{M} \sum_{m=1}^{M} \sigma_{m} h\left(x^{(m)}\right)\right].$$ Here, $\boldsymbol{\sigma}=\left[\sigma_{1}, \ldots, \sigma_{M}\right]^{\top}$ and $ \sigma_{m} $ is independent and uniformly distributed in $\{-1,+1\}$.

Furthermore, the hypothesis space generated by QML models $\mathcal{H}_Q$ is a subset of $\mathcal{H} = \{h(\boldsymbol{x}) = \boldsymbol{w}^{\top} \boldsymbol{x} : \|\boldsymbol{w}\|_2=1,|\boldsymbol{w}^{\top}\boldsymbol{x}| \leqslant 1 \}$, since $\boldsymbol{w} = T ^{\top} \boldsymbol{m}$ in $\mathcal{H}_Q$ requires that $T$ corresponds to a realizable quantum circuit. We prove that the Rademacher complexity of $\mathcal{H}$ with respect to dataset $S$ with $M$ samples is $\widehat{\mathfrak{R}}_{S}(\mathcal{H}) \leqslant \sqrt{1 / M}$ and $\widehat{\mathfrak{R}}_{S}(\mathcal{H}_Q) \leqslant \widehat{\mathfrak{R}}_{S}(\mathcal{H}) \leqslant \sqrt{1 / M}$ (see SM.~\ref{asec:generalization_bound}). Thus, we can derive the generalization bound for QML directly. The following theorem formally states this result.

\begin{theorem}
  \label{thm:generalization_bound}
  Let \(\mathcal{D}\) be a data distribution over \(\mathcal{X} \times \mathcal{Y}\), and let \(S = \{(\boldsymbol{\alpha}^{(m)}, y^{(m)})\}_{m=1}^M\) be a dataset of \(M\) independent and identically distributed (i.i.d.) samples drawn from \(\mathcal{D}\). Let the observable $O$ be a Pauli string with spectral norm $B_O$. Consider a QML model trained on $S$ with respect to the observable $O$, which produces a hypothesis $h_S \in \mathcal{H}_Q$. Assume the non-negative risk function $r: \mathcal{Y} \times \mathcal{Y} \rightarrow \mathbb{R}$ is uniformly bounded by $C > 0$ and is $L$-Lipschitz in its first variable for any fixed $y \in \mathcal{Y}$. Then, with probability at least $1 - \delta$ over the random sampling of $S$, the generalization error of $h_S$ satisfies:
   \begin{equation*}
   \begin{aligned}
     \operatorname{gen}(h_S)  \leqslant 2 L B_O\sqrt{\frac{1}{M}} +3C \sqrt{\frac{\log \frac{2}{\delta}}{2 M}}.
   \end{aligned}
   \end{equation*}
 \end{theorem}
 This generalization upper bound holds with probability at least $1 - \delta$. That is, among all possible datasets composed of $M$ quantum states and their labels, if we randomly sample a training set $S$ to train a hypothesis $h_S$, the probability that its generalization error does not exceed this bound is at least $1 - \delta$. The detailed proof is provided in SM.~\ref{asec:generalization_bound}.

Specifically, in regression problems, when the target function $f(\boldsymbol{x}) \in [0,1]$, we choose the observable $O_Z= \bigotimes_{n=1}^{N} Z_n$ with $B_O = 1$, and the per-sample risk function is absolute function $r(h_S(\boldsymbol{\alpha}),y) = |h_S(\boldsymbol{\alpha}) - y|$, which has an upper bound of $C = 1$ and a Lipschitz constant of $L = 1$. Then, the generalization bound satisfies the following inequality with probability at least $1-\delta$:
\begin{equation}
  \label{eq:regression_bound}
\begin{aligned}
  \operatorname{gen}(h_S) \leqslant \frac{2}{\sqrt{M}} + 3\sqrt{\frac{\log \frac{2}{\delta} }{2M}}.
\end{aligned}
\end{equation}

In binary classification problems with labels $y \in \{-1,1\}$, the 0-1 risk function $r(h_S(\boldsymbol{\alpha}),y) = \mathbbm{1}(h_S(\boldsymbol{\alpha}) \neq y)$ is commonly used, where $r(h_S(\boldsymbol{\alpha}),y) = 1$ when $h_S(\boldsymbol{\alpha}) \neq y$, and $r(h_S(\boldsymbol{\alpha}),y) = 0$ otherwise, with an obvious bound $C=1$. Although this 0-1 risk function itself is not Lipschitz continuous, it can be treated as  $L= 1 / 2$ for generalization bound analysis purposes (see SM.~\ref{asec:generalization_bound}). Therefore, by choosing the observable $O = Z_{1}$ with $B_O = 1$, the generalization bound satisfies the following inequality:
\begin{equation}
  \label{eq:classification_bound}
\begin{aligned}
  \operatorname{gen}(h_S) \leqslant \frac{1}{\sqrt{M}} + 3\sqrt{\frac{\log \frac{2}{\delta} }{2M}}.
\end{aligned}
\end{equation}
For $K$ classification problems, we can solve the problem by decomposing it into $K$ binary classification problems, where each binary classification problem selects one class as the positive class and the rest as the negative class. The generalization bound for the $K$ classification problem is:
$$
\begin{aligned}
\operatorname{gen}(h_S) \leqslant \frac{K}{\sqrt{M}} + 3K \sqrt{\frac{\log \frac{2}{\delta} }{2M}}.
\end{aligned}
$$

\noindent\textbf{\textit{Numerical Experiments}.--} To verify the effectiveness of our generalization upper bound, we conduct quantum phase classification on the axial next-nearest-neighbor Ising (ANNNI) model~\cite{wang2024supervised}. (see SM.~\ref{asec:experiments_details} for more details.)

\begin{figure}[htpb]
  \includegraphics[width=0.48\textwidth]{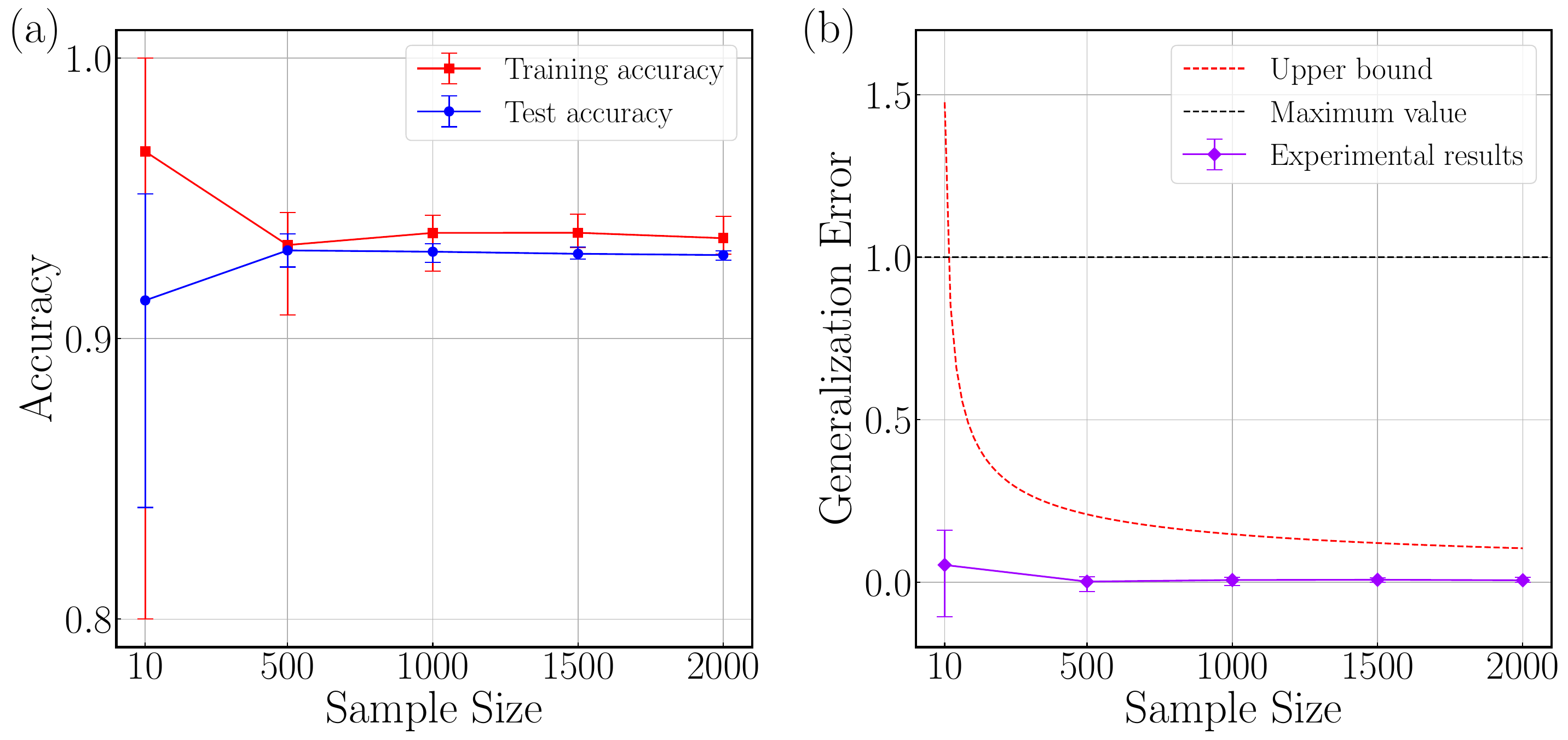}
  \caption{(a) Training accuracy and test accuracy under different sample sizes; (b) Comparison between experimental generalization error and theoretical generalization upper bound with confidence  $1-\delta = 0.9$. The maximum possible generalization error is 1. The error bars represent the minimum and maximum values across 10 independent runs with different training sets, with the central line showing the mean value.}
  \label{fig:gen_bound_classification}
\end{figure}

We consider an ANNNI model with $N$ qubits, whose Hamiltonian is given by $$H=-\left(\sum_{i=1}^{N-1} X_{i} X_{i+1}-\kappa \sum_{i=1}^{N-2} X_{i} X_{i+2}+h \sum_{i=1}^{N} Z_{i}\right).$$ The ground state of this Hamiltonian can exist in either ordered phase or disordered phase depending on the values of $\kappa$ and $h$ and two boundary lines $h_I(\kappa) \approx  \frac{1-\kappa}{\kappa}\left(1-\sqrt{\frac{1-3 \kappa+4 \kappa^{2}}{1-\kappa}}\right)$ and $h_{C}(\kappa) \approx 1.05 \sqrt{(\kappa-0.5)(\kappa-0.1)}$, which serve as the labels for the ground states. The 6-qubit phase diagram is shown in Fig.~\ref{fig:experiments}.(a), and we use the quantum circuit shown in Fig.~\ref{fig:experiments}.(b) for our experiments, where $L = 20$. Since the probabilistic nature of the generalization bound stems from the sampling of training set $S$, we select different training sets while keeping the same test set (with 10,000 test samples to ensure test error approximates prediction error). The training and test results, along with the generalization bound as a function of sample size, are shown in Fig.~\ref{fig:gen_bound_classification}.

In this experiment, we choose the per-sample risk function as the 0-1 risk function. The training accuracy thus is given by $1 - \hat{R}_S(h_S)$. In this scenario, our proposed generalization bound is  shown in Eq.~\eqref{eq:classification_bound}. Therefore, in Fig.~\ref{fig:gen_bound_classification}, we present the training and test results in terms of accuracy and calculate the generalization error based on these results. We can observe that our generalization bound is below the maximum possible generalization error of 1, and as the sample size increases, both the generalization upper bound and experimental results approach 0, demonstrating the tightness of our generalization bound.
\begin{figure}[htpb]
  \centering
  \includegraphics[width=0.48\textwidth]{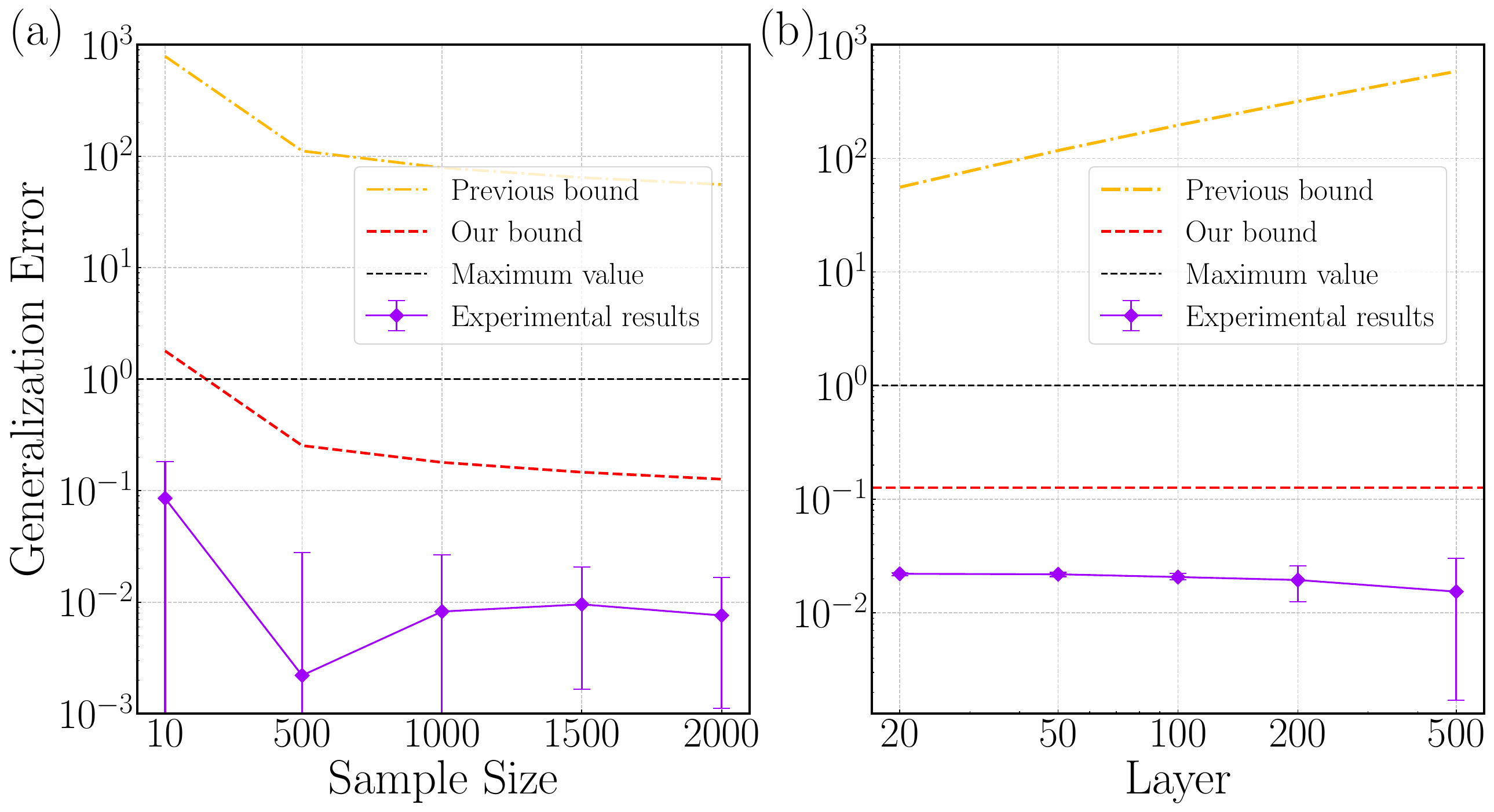}
  \caption{(a) Comparison between our theoretical generalization upper bound and previous work~\cite{caro2022generalization}. Both bounds are shown with confidence $1-\delta = 0.9$. Some minimal experimental results with negative generalization error are not displayed due to the logarithmic scale.  (b) Comparison of theoretical generalization error upper bounds under different model complexities, where layers represent complexity.  The error bars represent the minimum and maximum values across 10 independent runs with different training sets or random seeds, with the central line showing the mean value. }
  \label{fig:gen_bound_compare}
\end{figure}

An intuitive notion suggests that the generalization bound of QML models should increase with the number of parameterized quantum gates. Paper~\cite{caro2022generalization} also derived a generalization bound that increases with the number of parameterized quantum gates (see SM.~\ref{asec:recap}). However, comparing the experimental generalization error from Fig.~\ref{fig:gen_bound_classification}.(b), our theoretical bound, and their results in Fig.~\ref{fig:gen_bound_compare}.(a), we observe that even when the sample size $M$ increases to 2000, the generalization upper bound from~\cite{caro2022generalization} remains larger than the maximum possible generalization error of 1, meaning their bound provides no useful information. As we  increase the number of layers in the QML model (in layers $\{20,50,100,200,500\}$, the corresponding number of parameterized quantum gates is $\{360,900,1800,3600,9000\}$),  as shown in Fig.~\ref{fig:gen_bound_compare}.(b), we find that the experimental generalization error changes little, always staying below our theoretical bound, and potentially achieving even smaller generalization error, indicating better generalization capability.

Furthermore, the work~\cite{gil2024understanding} points out that when labels are completely random (i.e., labels are independent of data), QML can still train well but fail in prediction, suggesting that any uniform generalization bound may not be the correct way to measure generalization. However, we note that their experimental results are based on extremely small datasets, while generalization bounds are only meaningful when sufficient data is available. Since the labels and data are completely random, the prediction error remains at the random guessing level (around 0.5). However, because the labels are random, the expectations for different classes should be the same. Paper~\cite{wang2025limitations} shows that in this case, as the amount of data increases, the training error gradually increases. Therefore, as the amount of data increases, the training error gradually increases, ultimately leading to a small generalization error, as shown in Fig.~\ref{fig:random_label}.

\begin{figure}[htpb]
  \centering
  \includegraphics[width=0.48\textwidth]{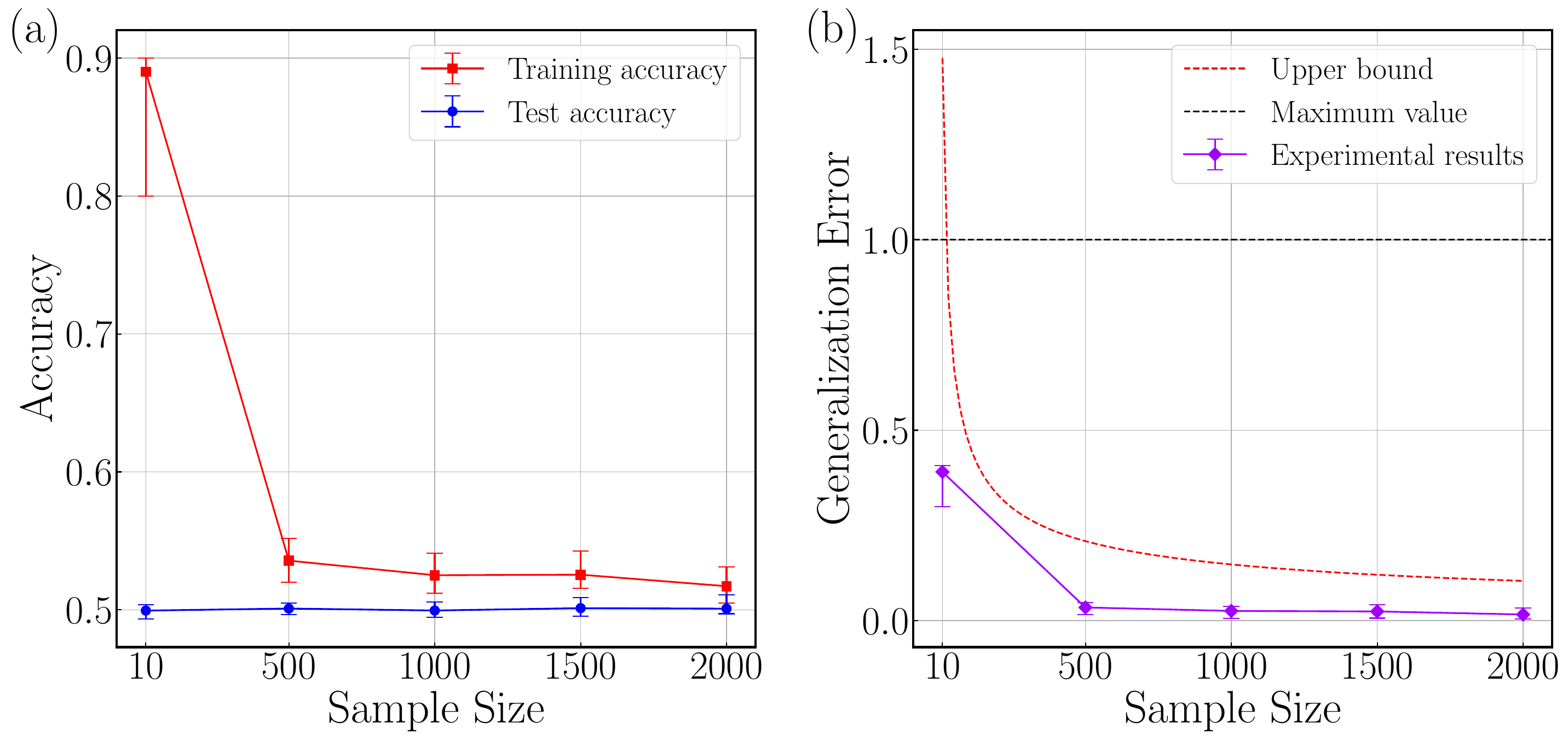}
  \caption{(a) Training accuracy and test accuracy under random labels. (b) Comparison between experimental generalization error and our generalization upper bound under random labels. The error bars represent the minimum and maximum values across 10 independent runs with different training sets, with the central line showing the mean value.}
  \label{fig:random_label}
\end{figure}

Besides, some works also consider that encoding methods~\cite{caro2021encodingdependent} and optimization approaches~\cite{yang2025stability} can affect generalization error, and derive generalization bounds, but they all involve big-O notation. We also validate this generalization upper bound on regression tasks and demonstrate that data dimension, number of qubits, and encoding methods have little impact on generalization error, as detailed in SM.~\ref{asec:regression_experiments}. We further verify that the choice of training hyperparameters, batch size, epochs, learning rate, and optimizer also have minimal effect on generalization error, as shown in SM.~\ref{asec:optimization}.

\noindent\textbf{\textit{Discussion}.-- }In this paper, we derived tight generalization bounds for QML, challenging the widely held view that generalization capability decreases as the number of parameterized quantum gates increases, and demonstrated that our bound remains valid even in random label scenarios. We emphasize that meaningful generalization bounds need not involve large constants or obscure significant variables within big-O notation, deriving truly meaningful bounds helps us better understand QML. 

Additionally, this work considers only the case where encoding gates and parameterized gates are separated. Another QML paradigm, data re-uploading, where the data encoding and variational parameterized quantum circuits are interleaved~\cite{perez2020data}, is not applicable to our analysis. Moreover, recent work has shown that this paradigm, when using deep circuits to process high-dimensional data, always approaches random guessing regardless of training quality, losing the good generalization properties of QML~\cite{wangpredictive}.

\let\oldaddcontentsline\addcontentsline
\renewcommand{\addcontentsline}[3]{}
\bibliography{ref}
\let\addcontentsline\oldaddcontentsline

\onecolumngrid

\pagebreak

\part*{Supplementary Material}

\appendix

\phantomsection
\tableofcontents

\resetAppendixCounters{A}

\section{Pauli Basis Representation}
\label{asec:pauli_basis}

Consider the density matrix $\rho$ of an $N$-qubit quantum state, which can be decomposed in the Pauli basis as:
\begin{equation}
  \label{aeq:pauli}
\begin{aligned}
\rho = \frac{1}{2^{N}} \left( \sum_{P_i \in \{I,Z,Y,X\}^{\otimes N}}^{} \alpha_i P_i \right) ,
\end{aligned}
\end{equation}
where $\alpha_i = \operatorname{Tr}\left[\rho P_i\right]$ is the coefficient of the Pauli basis $P_i$. Since $P_i$ is a Hermitian matrix, $\alpha_i$ must be a real number. The summation above contains $4^N$ terms. Therefore, the quantum state $\rho$ can be represented as a coefficient vector $\boldsymbol{\alpha} = \begin{bmatrix} 
    \alpha_1 & \cdots & \alpha_{4^N} 
\end{bmatrix}^{\top}$ determined by $4^N$ Pauli coefficients $\alpha_i$. Throughout the analysis, we adopt the Pauli basis ordering: $\{I,Z,X,Y\}$ for single-qubit systems and its $N$-fold tensor product $\{I,Z,X,Y\}^{\otimes N}$ for $N$-qubit systems.



For the general observable $O$, it can also be decomposed in unnormalized the Pauli basis as:
\begin{equation}
  \label{aeq:O_decomposition}
\begin{aligned}
  O = B_O \left(  \sum_{P_i \in \{I,Z,Y,X\}^{\otimes N}}^{} m_i P_i \right) ,
\end{aligned}
\end{equation}
where $B_O = \|O\|_{2}$ is the spectral norm of the observable $O$ and coefficient $m_i = \operatorname{Tr}\left[O P_i\right] /B_O$. The vector $\boldsymbol{m} = [m_1, \cdots, m_{4^N}]^{\top}$ represents the coefficient vector of the observable $O$ in the Pauli basis.

After converting both the quantum state and observable into their equivalent Pauli basis coefficient vectors, the expectation value $\operatorname{Tr}\left[O\rho\right]$ can then be expressed as the dot product of the observable and quantum state coefficient vectors, as demonstrated in the following theorem.

\begin{theorem}
  \label{athm:pauli_basis_representation}
  Let $\boldsymbol{\alpha}$ be the Pauli basis coefficient vector corresponding to an $N$-qubit quantum state $\rho$ defined in Eq.~\eqref{aeq:pauli}, and $\boldsymbol{m}$ be the Pauli coefficient vector corresponding to the observable $O $ with spectral norm $B_O$ defined in Eq.~\eqref{aeq:O_decomposition}. Then the measurement result of the quantum state $\rho$ with respect to the observable $O$, $\operatorname{Tr}\left[O \rho\right]$, can be expressed as:
  $$
  \begin{aligned}
  \operatorname{Tr}\left[O \rho\right] =  B_O \cdot \boldsymbol{m}^{\top} \boldsymbol{\alpha}.
  \end{aligned}
  $$
\end{theorem}

\begin{proof}
  Since
  $$
  \begin{aligned}
  \rho =  \frac{1}{2^{N}} \left( \sum_{P_i \in \{I,Z,Y,X\}^{\otimes N}}^{} \alpha_i P_i \right), O = B_O \left(  \sum_{P_i \in \{I,Z,Y,X\}^{\otimes N}}^{} m_i P_i \right) ,
  \end{aligned}
  $$
  we have 
  $$
  \begin{aligned}
  \operatorname{Tr}\left[O \rho\right] &=\frac{B_O}{2^{N}} \operatorname{Tr}\left[ \left( \sum_{P_i \in \{I,Z,Y,X\}^{\otimes N}}^{} m_i P_i \right) \cdot \left(  \sum_{P_j \in \{I,Z,Y,X\}^{\otimes N}}^{} m_j P_j \right)  \right] \\
  &= \frac{B_O}{2^{N}} \left( \sum_{P_i \in \{I,Z,Y,X\}^{\otimes N}} \alpha_i m_i \operatorname{Tr}\left[P_i^2\right] + \sum_{P_k \in \{I,Z,Y,X\}^{\otimes N},P_k \neq P_j} \alpha_k m_j \operatorname{Tr}\left[P_k P_j\right] \right)  \\
  &= B_O\left( \sum_{i=1}^{4^{N}} a_i m_i \right)  \\
  &= B_O \cdot \boldsymbol{m}^{\top} \boldsymbol{\alpha}, \\
  \end{aligned}
  $$
  where we used the properties of Pauli bases: $\operatorname{Tr}\left[ P_i^2 \right] = 2^{N}$ and  $\operatorname{Tr}\left[P_i P_j \right] = 0, \forall i \neq j$. Thus, $  \operatorname{Tr}\left[O \rho\right] =  B_O \cdot \boldsymbol{m}^{\top}  \boldsymbol{\alpha}.$
\end{proof}

In quantum machine learning, we generally consider an observable as a single Pauli string, meaning the observable can be expressed as $O = B_O \cdot m_j P_j$ for $j \in [1:4^{N}]$, which represents any integer from $1$ to $4^{N}$. Therefore, for such an observable, the corresponding Pauli coefficient vector can be written as:
\begin{equation}
\label{eq:Om}
\begin{aligned}
\boldsymbol{m} = \begin{bmatrix} 
     0 & \cdots & \underset{j\text{-th}}{1} & \cdots & 0
\end{bmatrix}^{\top} .
\end{aligned}
\end{equation}
And we have the following corollary:

\begin{corollary}
  \label{acor:measure_properties}
  Consider an observable as a single Pauli string  $O = B_O \cdot m_j P_j$ with spectral norm $B_O$, whose coefficient vector $\boldsymbol{m}$ is defined as in Eq.~\eqref{eq:Om}. For any quantum state $\rho$ with corresponding Pauli basis coefficient vector $\boldsymbol{\alpha}$, we have
  $$
  \begin{aligned}
  |\boldsymbol{m}^{\top} \boldsymbol{\alpha}| \leqslant 1.
  \end{aligned}
  $$
\end{corollary}

\begin{proof}
  Since $O = B_O \cdot m_j P_j$, according to Theorem~\ref{athm:pauli_basis_representation}, we have $\operatorname{Tr}\left[O \rho\right] = B_O \cdot m_j \alpha_j$. Then $|\boldsymbol{m}^{\top} \boldsymbol{\alpha}| = |m_j \alpha_j| = \frac{1}{B_O}|\operatorname{Tr}\left[O \rho\right]| \leqslant \frac{1}{B_O} \cdot B_O = 1$, where we used the fact that the expectation value of any observable is bounded by its spectral norm. Thus, $|\boldsymbol{m}^{\top} \boldsymbol{\alpha}| \leqslant 1$.
\end{proof}

Next, we introduce the properties of quantum circuits in the Pauli basis:
\begin{lemma}[Theorem C.1 in Ref.~\cite{wangpredictive}]
  For an $N$-qubit parameterized quantum circuit $\mathcal{U}_{\theta}$, the corresponding Pauli basis transfer matrix $T(\boldsymbol{\theta})$ is orthogonal.
\end{lemma}
This indicates that the transfer matrix $T(\boldsymbol{\theta})$ of  parameterized quantum circuit $\mathcal{U}_{\boldsymbol{\theta}}(\rho) = U(\boldsymbol{\theta}) \rho U(\boldsymbol{\theta})^{\dagger}$ used in quantum machine learning models has the following properties in the Pauli basis: $T(\boldsymbol{\theta})^{\top} T(\boldsymbol{\theta}) = I$ and its spectral norm $\|T(\boldsymbol{\theta})\|_{2} = 1$. It is worth noting that $T(\boldsymbol{\theta})$ contains two aspects of information about the parameterized quantum circuit: the circuit architecture is primarily reflected in the matrix form $T$, while the circuit parameters are embodied in $\boldsymbol{\theta}$. Moreover, while the Pauli basis transfer matrix corresponding to a parameterized quantum circuit is orthogonal, this does not imply that all orthogonal matrices in Pauli basis can correspond to parameterized quantum circuits.

In summary, for a quantum state $\rho$, the output of a quantum machine learning model implemented by a parameterized quantum circuit $\mathcal{U}_{\boldsymbol{\theta}}$ and measured by a fixed single Pauli string observable $O$ can be represented as:

\begin{equation*}
\begin{aligned}
  h(\rho,\boldsymbol{\theta}) = h(\boldsymbol{\alpha},\boldsymbol{\theta}) = \operatorname{Tr}\left[O \mathcal{U}_{\boldsymbol{\theta}}( \rho)\right] = B_O \cdot \boldsymbol{m}^{\top} T(\boldsymbol{\theta} ) \boldsymbol{\alpha} = B_O \cdot \boldsymbol{w}^{\top}(T,\boldsymbol{\theta}) \boldsymbol{\alpha} .\\
\end{aligned}
\end{equation*}
Here, for a fixed observable coefficient vector $\boldsymbol{m}$, the parameter vector of the quantum machine learning model $\boldsymbol{w}(T,\boldsymbol{\theta}) = T ^{\top}(\boldsymbol{\theta}) \boldsymbol{m}$ is a function of the parameterized quantum circuit architecture $T$ and circuit parameters $\boldsymbol{\theta}$, with $\|\boldsymbol{w}(T,\boldsymbol{\theta})\|_{2}^2 = \boldsymbol{m}^{\top} T (\boldsymbol{\theta}) T ^{\top}(\boldsymbol{\theta}) \boldsymbol{m} = \boldsymbol{m}^{\top} \boldsymbol{m} = 1$. Moreover, according to Corollary~\ref{acor:measure_properties}, $|\boldsymbol{w}^{\top}(T,\boldsymbol{\theta}) \boldsymbol{\alpha}| \leqslant 1$.  Therefore, when $B_O = 1$, for input $\boldsymbol{x} \in \mathcal{X}$, the hypothesis set $\mathcal{H}_Q$ generated by the quantum machine learning model is:
\begin{equation}
  \label{eq:HQ}
\begin{aligned}
  \mathcal{H}_Q = \{h(\boldsymbol{x}) =  \boldsymbol{w}^{\top} \boldsymbol{x}: \|\boldsymbol{w}\|_{2} = 1,|\boldsymbol{w}^{\top} \boldsymbol{x}| \leqslant 1 \  \forall  \  \boldsymbol{x} \in \mathcal{X},  \boldsymbol{w} \text{ satisfies certain special structure}\},
\end{aligned}
\end{equation}
where special structure means that $\boldsymbol{w} = T ^{\top} \boldsymbol{m}$ and the transfer matrix $T$ in Pauli basis corresponds to a quantum circuit. Clearly, the hypothesis set generated by the quantum machine learning model is a subset of the following hypothesis set $\mathcal{H}$:
\begin{equation}
  \label{eq:HH}
\begin{aligned}
  \mathcal{H} = \{h(\boldsymbol{x}) = \boldsymbol{w}^{\top} \boldsymbol{x} : \|\boldsymbol{w}\|_{2} = 1,|\boldsymbol{w}^{\top} \boldsymbol{x}| \leqslant 1 \  \forall  \  \boldsymbol{x} \in \mathcal{X}\},
\end{aligned}
\end{equation}
namely, $\mathcal{H}_Q \subseteq \mathcal{H}$.

\resetAppendixCounters{B}

\section{Generalization Bound for QML}
\label{asec:generalization_bound}
We first introduce Rademacher complexity, then present generalization error bounds based on Rademacher complexity, and finally derive generalization bounds for quantum machine learning. 

Let us introduce some notation. Let $z = (x,y) \in \mathcal{Z}$ represent a sample containing both data and label, where the sample space $\mathcal{Z} = \mathcal{X} \times \mathcal{Y}$ consists of input space $\mathcal{X}$ and label space $\mathcal{Y}$. The hypothesis function $h \in \mathcal{H}$ generated by the machine learning model maps from input space to label space, i.e., $h: \mathcal{X} \rightarrow \mathcal{Y}$, where $\mathcal{H}$ is the hypothesis space. The per-sample risk function is $r: \mathcal{Y} \times \mathcal{Y} \rightarrow \mathbb{R}$. We define the composite function $g = r \circ h: \mathcal{Z} \to \mathbb{R}$ as $g(z) = g((x,y)) = r(h(x),y) = r(y',y)$, where $y' = h(x)$ is the predicted label. Based on this notation, we define the training error as:
$$
\begin{aligned}
\widehat{R}_S(h) = \frac{1}{M} \sum_{m=1}^{M}r(h(x^{(m)}),y^{(m)}) =\frac{1}{M} \sum_{m=1}^{M} g((x^{(m)},y^{(m)})) = \frac{1}{M} \sum_{m=1}^{M}g(z^{(m)}),
\end{aligned}
$$
similarly, we define the prediction error as:
$$
\begin{aligned}
R(h) = \underset{(x,y) \sim \mathcal{D}}{\mathbb{E}}[r(h(x),y)] = \underset{(x,y) \sim \mathcal{D}}{\mathbb{E}}[g((x,y))] = \underset{z \sim \mathcal{D}}{\mathbb{E}}[g(z)].
\end{aligned}
$$

\begin{definition}[Empirical Rademacher Complexity]
    Let $\mathcal{G}$ be a family of functions mapping from sample space $Z = \mathcal{X} \times \mathcal{Y}$ to $[a, b]$, and $S=\left(z^{(1)}, \cdots, z^{(M)}\right)$ be a fixed dataset of size $M$, where $\mathcal{X}$ is the input space, $\mathcal{Y}$ is the label space, and $z^{(m)} = (x^{(m)},y^{(m)})$. The empirical Rademacher complexity of $\mathcal{G}$ with respect to dataset $S$ is defined as:
  $$
  \begin{aligned}
  \widehat{\mathfrak{R}}_{S}(\mathcal{G})=\underset{\boldsymbol{\sigma}}{\mathbb{E}}\left[\sup _{g \in \mathcal{G}} \frac{1}{M} \sum_{m=1}^{M} \sigma_{m} g\left(z^{(m)}\right)\right],
  \end{aligned}
  $$
where $\boldsymbol{\sigma}=\left[\sigma_{1}, \ldots, \sigma_{M}\right]^{\top}$, and $ \sigma_{m} $ is independent uniform random variables taking values in \( \{-1,+1\} \). These random variables \( \sigma_{m} \) are called Rademacher variables.
\end{definition}

\begin{lemma}[Theorem 3.3 in Ref.~\cite{mohri2018foundations}]
  \label{lem:generalization_Rademacher}
    Let $ \mathcal{G} $ be a family of functions mapping from $ \mathcal{Z} $ to $ [0,C] $. Then, for any $ \delta>0 $, with probability at least $ 1-\delta $ over the draw of an i.i.d. sample $ S $ of size $ M $, the following inequality holds for all $ g \in \mathcal{G} $:
  \begin{equation*}
  \begin{aligned}
    \mathbb{E}[g(z)] \leqslant \frac{1}{M} \sum_{m=1}^{M} g\left(z^{(m)}\right)+2 \widehat{\mathfrak{R}}_{M}(\mathcal{G})+3C\sqrt{\frac{\log \frac{2}{\delta}}{2 M}}.
  \end{aligned}
  \end{equation*}
\end{lemma}

 \begin{lemma}[Talagrand's lemma, Lemma 5.7 in Ref.~\cite{mohri2018foundations}]
   \label{lem:generalization_Rademacher_Lipschitz}
   Let $r: \mathcal{Y} \times \mathcal{Y} \rightarrow \mathbb{R}$ be an $L$-Lipschitz function with respect to its first variable. For any datasets $ S= \{(x^{(m)},y^{(m)})\}_{m=1}^{M} $, let $ S_{\mathcal{X}} $ denote its projection over $ \mathcal{X}: S_{\mathcal{X}}= \{(x^{(m)})\}_{m=1}^{M} $. Then, the following relation holds between the empirical Rademacher complexities of $ \mathcal{G} $ and $ \mathcal{H} $:
  $$
  \begin{aligned}
    \widehat{\mathfrak{R}}_{S}(\mathcal{G}) = \widehat{\mathfrak{R}}_{S}(r \circ \mathcal{H}) \leqslant L \widehat{\mathfrak{R}}_{S_{\mathcal{X}}}(\mathcal{H}) .
  \end{aligned}
  $$
 \end{lemma}

\begin{lemma}
  \label{lem:Rademacher_compare}
  Consider two hypothesis spaces $\mathcal{H}_1$ and $\mathcal{H}_2$. If $\mathcal{H}_1 \subseteq \mathcal{H}_2$, then for any dataset $S_{\mathcal{X}} = \{(\boldsymbol{x}^{(m)})\}_{m=1}^{M}$, we have $\widehat{\mathfrak{R}}_{S_{\mathcal{X}}}(\mathcal{H}_1) \leqslant \widehat{\mathfrak{R}}_{S_{\mathcal{X}}}(\mathcal{H}_2)$.
\end{lemma}
\begin{proof}
  Since $\mathcal{H}_1 \subseteq \mathcal{H}_2$, for any dataset $S_{\mathcal{X}} = \{(\boldsymbol{x}^{(m)})\}_{m=1}^{M}$, we have
  $$\sup_{h \in \mathcal{H}_1} \sum_{m=1}^{M} \sigma_m h(\boldsymbol{x}^{(m)}) \leqslant \sup_{h \in \mathcal{H}_2} \sum_{m=1}^{M} \sigma_m h(\boldsymbol{x}^{(m)}).$$
  Therefore, $\widehat{\mathfrak{R}}_{S_{\mathcal{X}}}(\mathcal{H}_1) \leqslant \widehat{\mathfrak{R}}_{S_{\mathcal{X}}}(\mathcal{H}_2).$
\end{proof}

\begin{lemma}
  \label{alem:compact_set}
  For any $d$-dimensional vector $\boldsymbol{x} \in \mathcal{X} \subseteq \mathbb{R}^{d}$ taken from the input space $\mathcal{X}$, the set $\Omega = \{\boldsymbol{w}: \|\boldsymbol{w}\|_{2} = 1, |\boldsymbol{w}^{\top} \boldsymbol{x} | \leqslant 1, \forall \boldsymbol{x} \in \mathcal{X}\} \subseteq \mathbb{R}^{d}$ is a non-empty, bounded, closed set.
\end{lemma}
\begin{proof}
  The set $\Omega$ can be written as the intersection of the following sets:
  $$\Omega = \{\boldsymbol{w}: \|\boldsymbol{w}\|_{2} = 1\} \cap \bigcap_{\boldsymbol{x} \in \mathcal{X}} \{\boldsymbol{w}: |\boldsymbol{w}^{\top} \boldsymbol{x}| \leqslant 1\}.$$
  Here, $\{\boldsymbol{w}: \|\boldsymbol{w}\|_{2} = 1\}$ is the unit sphere, which is a closed  set. For each fixed $\boldsymbol{x} \in \mathcal{X}$, the set $\{\boldsymbol{w}: |\boldsymbol{w}^{\top} \boldsymbol{x}| \leqslant 1\} = \{\boldsymbol{w}: \boldsymbol{w}^{\top} \boldsymbol{x} \leqslant 1\} \cap  \{\boldsymbol{w}: \boldsymbol{w}^{\top} \boldsymbol{x} \geqslant -1\}$ is the intersection of two half-spaces, which is also a closed set. Since the intersection of any collection of closed sets remains closed, the set $\Omega$ is closed.  Clearly, for any $\boldsymbol{w} \in \Omega$, we have $\|\boldsymbol{w}\|_{2} = 1$, so $\Omega$ is bounded.
\end{proof}

\begin{lemma}
  \label{alem:extreme_value_theorem}
  For any non-empty, bounded, closed set $\Omega \subseteq \mathbb{R}^{d}$ and any given vector $\boldsymbol{s} \in \mathbb{R}^{d}$, there exists $\boldsymbol{w}^{*} \in \Omega$ such that:
  $$
  \begin{aligned}
  \sup_{\boldsymbol{w} \in \Omega} \boldsymbol{w}^{\top} \boldsymbol{s} = (\boldsymbol{w}^{*})^{\top} \boldsymbol{s}.
  \end{aligned}
  $$
\end{lemma}

\begin{proof}
  Since $\Omega$ is a non-empty, bounded, closed set in $\mathbb{R}^{d}$, by the Heine-Borel Theorem, $\Omega$ is compact. Moreover, since the function $\boldsymbol{w} \mapsto \boldsymbol{w}^{\top} \boldsymbol{s}$ is continuous on $\Omega$, by the Extreme Value Theorem~\cite{rudin1976principles}, there exists $\boldsymbol{w}^{*} \in \Omega$ such that $\sup_{\boldsymbol{w} \in \Omega} \boldsymbol{w}^{\top} \boldsymbol{s} = (\boldsymbol{w}^{*})^{\top} \boldsymbol{s}$.

\end{proof}

The following theorem demonstrates the Rademacher complexity of hypothesis space $\mathcal{H} = \{h(\boldsymbol{x}) = \boldsymbol{w}^{\top} \boldsymbol{x}: \boldsymbol{w} \in \Omega \}$:

\begin{theorem}
  \label{athm:Rademacher_quantum}
  For a dataset $S = \{(\boldsymbol{x}^{(m)},y^{(m)})\}_{m=1}^{M}$ of size $M$, let $S_{\mathcal{X}} = \{(\boldsymbol{x}^{(m)})\}_{m=1}^{M}$ be the projection of $S$ over $\mathcal{X}$. The empirical Rademacher complexity of the hypothesis space $\mathcal{H} = \{h(\boldsymbol{x}) = \boldsymbol{w}^{\top} \boldsymbol{x}: \boldsymbol{w} \in \Omega \}$ satisfies
  $$
  \begin{aligned}
    \widehat{\mathfrak{R}}_{S_{\mathcal{X}}}(\mathcal{H}) \leqslant \sqrt{\frac{1}{M}},
  \end{aligned}
  $$
  where $\Omega = \{\boldsymbol{w}: \|\boldsymbol{w}\|_{2} = 1, |\boldsymbol{w}^{\top} \boldsymbol{x}| \leqslant 1, \forall  \  \boldsymbol{x} \in \mathcal{X}\}$.
\end{theorem}

\begin{proof}
\begin{equation*}
\begin{aligned}
    \widehat{\mathfrak{R}}_{S_{\mathcal{X}}}(\mathcal{H}) & =\frac{1}{M} \underset{\boldsymbol{\sigma}}{\mathbb{E}}\left[\sup_{h \in \mathcal{H}} \sum_{m=1}^{M} \sigma_{m} h(\boldsymbol{x}^{(m)}) \right]  \\
    & = \frac{1}{M} \underset{\boldsymbol{\sigma}}{\mathbb{E}}\left[\sup_{\boldsymbol{w} \in \Omega } \sum_{m=1}^{M} \sigma_m \boldsymbol{w}^{\top} \boldsymbol{x}^{(m)} \right] \\
    & = \frac{1}{M} \underset{\boldsymbol{\sigma}}{\mathbb{E}}\left[\sup_{\boldsymbol{w} \in \Omega } \boldsymbol{w}^{\top} \sum_{m=1}^{M} \sigma_m  \boldsymbol{x}^{(m)} \right]. \\
\end{aligned}
\end{equation*}
Let $\boldsymbol{s} = \sum_{m=1}^{M} \sigma_{m} \boldsymbol{x}^{(m)}$. According to  Lemma~\ref{alem:compact_set} and Lemma~\ref{alem:extreme_value_theorem}, for a fixed $\boldsymbol{\sigma}$, there exists $\boldsymbol{w}^{*}(\boldsymbol{\sigma}) \in \Omega$ such that 
$$
\begin{aligned}
  \sup_{\boldsymbol{w} \in \Omega} \boldsymbol{w}^{\top} \boldsymbol{s} = (\boldsymbol{w}^{*}(\boldsymbol{\sigma}))^{\top} \boldsymbol{s}.
\end{aligned}
$$
Therefore,
$$
\begin{aligned}
\widehat{\mathfrak{R}}_{S_{\mathcal{X}}}(\mathcal{H}) &= \frac{1}{M} \underset{\boldsymbol{\sigma}}{\mathbb{E}}\left[ \sum_{m=1}^{M} \sigma_m (\boldsymbol{w}^{*}(\boldsymbol{\sigma}))^{\top} \boldsymbol{x}^{(m)} \right] \\
& \leqslant  \frac{1}{M} \sqrt{\underset{\boldsymbol{\sigma}}{\mathbb{E}} \left[ \left(  \sum_{m=1}^{M} \sigma_m (\boldsymbol{w}^{*}(\boldsymbol{\sigma}))^{\top} \boldsymbol{x}^{(m)} \right)^2 \right] },  \\
\end{aligned}
$$
where the inequality follows from applying the Cauchy Schwarz inequality to the expectation, viewing  $Z = \sum_{m=1}^{M} \sigma_m (\boldsymbol{w}^{*}(\boldsymbol{\sigma}))^{\top} \boldsymbol{x}^{(m)}$ as a random variable and noting that $\mathbb{E}[Z] = \mathbb{E}[Z \cdot 1] \leqslant \sqrt{\mathbb{E}[Z^2] \mathbb{E}[1^2]} = \sqrt{\mathbb{E}[Z^2]}$. Thus

$$
\begin{aligned}
\underset{\boldsymbol{\sigma}}{\mathbb{E}} \left[ \left(  \sum_{m=1}^{M} \sigma_m (\boldsymbol{w}^{*}(\boldsymbol{\sigma}))^{\top} \boldsymbol{x}^{(m)} \right)^2  \right] &=  \underset{\boldsymbol{\sigma}}{\mathbb{E}} \left[ \sum_{m=1}^{M} \sum_{n=1}^{M} \sigma_m \sigma_n (\boldsymbol{w}^{*}(\boldsymbol{\sigma}))^{\top} \boldsymbol{x}^{(m)} (\boldsymbol{x}^{(n)})^{\top} \boldsymbol{w}^{*}(\boldsymbol{\sigma})\right]  \\
& \leqslant \underset{\boldsymbol{\sigma}}{\mathbb{E}} \left[ \sum_{m=1}^{M} \sum_{n=1}^{M} \sigma_m \sigma_n \right] \\
& = \sum_{m=1}^{M} \sum_{n=1}^{M} \underset{\boldsymbol{\sigma}}{\mathbb{E}}[\sigma_m \sigma_n] \\
&= M, \\  
\end{aligned}
$$
where the inequality follows from  $\boldsymbol{w}^{*}(\boldsymbol{\sigma}) \in \Omega$, which implies $|(\boldsymbol{w}^{*}(\boldsymbol{\sigma}))^{\top} \boldsymbol{x}^{(m)}| \leqslant 1$ and $|(\boldsymbol{w}^{*}(\boldsymbol{\sigma}))^{\top} \boldsymbol{x}^{(n)}| \leqslant 1$ for all $\boldsymbol{x}^{(m)}, \boldsymbol{x}^{(n)} \in \mathcal{X}$. The last equality uses the fact that $\underset{\boldsymbol{\sigma}}{\mathbb{E}}[\sigma_{i} \sigma_{j}] = 0$ for $i \neq j$ and $\underset{\boldsymbol{\sigma}}{\mathbb{E}}[\sigma_{i}^2] = 1$ since the Rademacher variables are independent. Combining these results:
\begin{equation*}
\begin{aligned}
    \widehat{\mathfrak{R}}_{S_{\mathcal{X}}}(\mathcal{H}) \leqslant \frac{1}{M} \sqrt{M } = \sqrt{\frac{1}{M}} .
\end{aligned}
\end{equation*}
\end{proof}

Finally, we can derive the generalization bound for quantum machine learning as follows:

\begin{theorem}[Theorem~\ref{thm:generalization_bound} in the main paper]
  Let \(\mathcal{D}\) be a data distribution over \(\mathcal{X} \times \mathcal{Y}\), and let \(S = \{(\boldsymbol{\alpha}^{(m)}, y^{(m)})\}_{m=1}^M\) be a dataset of \(M\) independent and identically distributed (i.i.d.) samples drawn from \(\mathcal{D}\). Let the observable $O$ be a Pauli string with spectral norm $B_O$. Consider a quantum machine learning model trained on $S$ with respect to the observable $O$, which produces a hypothesis $h_S \in \mathcal{H}_Q$. Assume the non-negative risk function $r: \mathcal{Y} \times \mathcal{Y} \rightarrow \mathbb{R}$ is uniformly bounded by $C > 0$ and is $L$-Lipschitz in its first variable for any fixed $y \in \mathcal{Y}$. Then, with probability at least $1 - \delta$ over the random sampling of $S$, the generalization error of $h_S$ satisfies:
  \begin{equation*}
  \begin{aligned}
    \operatorname{gen}(h_S)  \leqslant 2 L B_O\sqrt{\frac{1}{M}} +3C \sqrt{\frac{\log \frac{2}{\delta}}{2 M}}.
  \end{aligned}
  \end{equation*}
\end{theorem}

\begin{proof}
  Let $S_{\mathcal{X}} = \{(\boldsymbol{\alpha}^{(m)})\}_{m=1}^{M}$ be the projection of $S$ over $\mathcal{X}$. When not considering the spectral norm of the observable, the hypothesis set $\mathcal{H}_Q$ generated by the quantum machine learning model and the general hypothesis set $\mathcal{H}$ are defined as in Eq.~\eqref{eq:HQ} and Eq.~\eqref{eq:HH}, respectively. According to Lemma~\ref{lem:generalization_Rademacher} and Lemma~\ref{lem:generalization_Rademacher_Lipschitz}, we have:
  $$
  \begin{aligned}
   R(h_S) \leqslant \widehat{R}_S(h_S) + 2 L \mathfrak{\widehat{R}}_{S_{\mathcal{X}}}(\mathcal{H}_Q)+3C\sqrt{\frac{\log \frac{2}{\delta}}{2 M}}.
  \end{aligned}
  $$

   According to Theorem~\ref{athm:Rademacher_quantum}, the Rademacher complexity of the general hypothesis set $\mathcal{H}$ satisfies $ \widehat{\mathfrak{R}}_{S_{\mathcal{X}}}(\mathcal{H}) \leqslant \sqrt{\frac{1}{M}}$. Since $\mathcal{H}_Q \subseteq \mathcal{H}$, by Lemma~\ref{lem:Rademacher_compare}, we have $ \widehat{\mathfrak{R}}_{S_{\mathcal{X}}}(\mathcal{H}_Q) \leqslant  \widehat{\mathfrak{R}}_{S_{\mathcal{X}}}(\mathcal{H}) \leqslant \sqrt{\frac{1}{M}}$. Incorporating the spectral norm $B_O$ of the observable, the generalization error can be bounded as:
  $$
  \begin{aligned}
  \operatorname{gen}(h_S) = R(h_S) - \widehat{R}_S(h_S) \leqslant 2 L B_O \sqrt{\frac{1}{M}}+3C\sqrt{\frac{\log \frac{2}{\delta}}{2 M}}.
  \end{aligned}
  $$
\end{proof}

\begin{lemma}[Lemma 3.4 in Ref.~\cite{mohri2018foundations}]
  Let $ \mathcal{H} $ be a family of functions taking values in $ \{-1,+1\} $ and let $ \mathcal{G} $ be the family of risk functions associated to $ \mathcal{H} $ for the 0-1 risk: $ \mathcal{G}=\{(x, y) \mapsto \mathbbm{1}(h(x) \neq y): h \in \mathcal{H} \} $. For any datasets $ S= \{(x^{(m)},y^{(m)})\}_{m=1}^{M} $ of elements in $ \mathcal{X} \times\{-1,+1\} $, let $ S_{\mathcal{X}} $ denote its projection over $ \mathcal{X}: S_{\mathcal{X}}= \{(x^{(m)})\}_{m=1}^{M} $. Then, the following relation holds between the empirical Rademacher complexities of $ \mathcal{G} $ and $ \mathcal{H} $ :
$$
\widehat{\mathfrak{R}}_{S}(\mathcal{G})=\frac{1}{2} \widehat{\mathfrak{R}}_{S_{\mathcal{X}}}(\mathcal{H}) .
$$
\end{lemma}

\resetAppendixCounters{C}

\section{Recap of Previous Generalization Bound}
\label{asec:recap}
We will review previous generalization bounds in quantum machine learning and the key insights they attempt to reveal in this section.

\subsection{Effects of Model Complexity on Generalization}
\label{subsec:few_data}

The work by Caro et al.~\cite{caro2022generalization} suggests a fundamental connection between generalization bounds in quantum machine learning and model complexity, specifically quantified by the number of parameterized quantum gates within the quantum circuit. They introduce a per-sample risk function framework (Eq. (1) in their work), which is subject to the spectral norm constraint $\sup_{\rho,y}\|O_{\rho,y}^{\text{loss}}\|_{2} \leqslant B_O$:
\begin{equation}
  \label{eq:C_loss} 
\begin{aligned}
  r(h_S((\rho)),y) = \operatorname{Tr}\left[O_{\rho,y}^{\text{loss}} \mathcal{E}_{\boldsymbol{\theta}^{*}}(\rho)\right],
\end{aligned}
\end{equation}
where $\boldsymbol{\theta}^{*}$ is the optimal parameters of the quantum machine learning model $\mathcal{E}_{\boldsymbol{\theta}}$ learned on the training set $S$.

Their theoretical analysis culminates in a generalization bound (Appendix C, Eq. (C.63) in paper~\cite{caro2022generalization}) that explicitly depends on the circuit complexity:
\begin{equation}
  \label{eq:previous_bound}
\begin{aligned}
  \operatorname{gen}(h_S) \leqslant \frac{24 B_O}{\sqrt{M}} \sqrt{512 T} \cdot\left(\frac{1}{2} \sqrt{\log (6 T)}+\frac{1}{2} \sqrt{\log 2}-\frac{\sqrt{\pi}}{2} \operatorname{erf}(\sqrt{\log 2})-\frac{\sqrt{\pi}}{2}\right)+3 B_O \sqrt{\frac{2 \log (2 / \delta)}{M}}.
\end{aligned}
\end{equation}
Here, $M$ denotes the sample size of training set, $T$ represents the number of parameterized quantum gates, and $\operatorname{erf}(x)=\frac{2}{\sqrt{\pi}} \int_{0}^{x} \exp \left(-t^{2}\right) \mathrm{d} t$ is the error function. This result suggests a scaling relationship $\operatorname{gen}(h_S) \in \mathcal{\tilde{O} }\left( \sqrt{\frac{T}{M}} \right)$ (where $\mathcal{\tilde{O} }$ suppresses logarithmic factors), leading to the conclusion that increasing model complexity through more parameterized quantum gates inherently deteriorates generalization capability.

For a fair comparison, we convert the setting from paper~\cite{caro2022generalization} to an equivalent formulation in our framework. This is clearly feasible since our setting is more general than that work, and we compare them on the binary classification problem. We set $O_{\rho,y}^{\text{loss}} = I - \ket{y}\bra{y}$, that is, $O_{\rho,y}^{\text{loss}} = I-\ket{1}\bra{1} = \ket{0}\bra{0}$ when $y = 1$, and $O_{\rho,y}^{\text{loss}} = I-\ket{0}\bra{0} = \ket{1}\bra{1}$ when $y = -1$. We define $\rho_y = \ket{1}\bra{1}$ when $y = 1$ and $\rho_y = \ket{0}\bra{0}$ when $y = -1$. The risk function in Eq.~\eqref{eq:C_loss} can then be simplified to:
\begin{equation}
  \label{eq:C_loss_equivalent}
\begin{aligned}
  r(h_S((\rho)),y) = 1 - \operatorname{Tr}\left[\rho_y\mathcal{E}_{\boldsymbol{\theta}^{*}}((\rho))\right].
\end{aligned}
\end{equation}
Since this risk function has an upper bound $C=1$, Lipschitz coefficient $L=1$, and $B_O = 1$, according to our results, the generalization bound using this risk function is:
\begin{equation}
  \label{eq:our_bound}
\begin{aligned}
  \operatorname{gen}(h_S) \leqslant \frac{2}{\sqrt{M}} + 3\sqrt{\frac{\log (2 / \delta)}{2M}}.
\end{aligned}
\end{equation}

This generalization bound in paper~\cite{caro2022generalization} suggests that as model complexity increases with more parameterized quantum gates, the generalization error grows at least at a square root rate with respect to the number of parameterized quantum gates, implying that more gates lead to worse generalization. However, as demonstrated in Fig.~\ref{fig:gen_bound_compare}.(b), our experimental results verify that the generalization error does not increase with model complexity, and our proposed generalization bound is both tighter and more meaningful. Furthermore, for the risk function in Eq.~\eqref{eq:C_loss} or its equivalent form in Eq.~\eqref{eq:C_loss_equivalent}, the maximum possible generalization error is 1, while this generalization bound consistently exceeds the maximum possible generalization error.

\subsection{Effects of Encoding Method and Data Dimension on Generalization}

In paper~\cite{caro2021encodingdependent}, the authors argue that the generalization bound for quantum machine learning is related to the encoding method and data dimension. This paper~\cite{caro2021encodingdependent} also considers a general per-sample risk function $r$ as our work and points out that when the encoding gate is $k$-local, i.e., $U(x) = e^{-i x H_{k}}$ where the Hamiltonian $H$ has dimension $2^k$ and has as many distinct eigenvalues as possible (e.g., $2^k$ eigenvalues), specifically when $H$ has eigenvalues $1,3,9,\cdots,3^{k}-1$, according to Equations (100), (111), and (115) in \cite{caro2021encodingdependent}, the generalization bound is:
\begin{equation}
  \label{eq:gen_bound_encoding}
\begin{aligned}
  \operatorname{gen}(h_S) &\leqslant \frac{12LB_O}{\sqrt{m}} \sqrt{\left( \frac{2^{k}(2^{k} -1)}{2} + 1 \right)^{d}}\left[ \sqrt{\log \left(3 \cdot 2(2 \pi)^{\frac{d}{2}} \right)+\frac{1}{2} \log \left( \left( \frac{2^{k}(2^{k} -1)}{2} + 1 \right)^{d} \right)}  +\int_{0}^{\gamma_{0}} \sqrt{\log \left(\frac{2}{\beta}\right)} \mathrm{d} \beta\right] \\
  &  + 3C \sqrt{\frac{\log (2 / \delta)}{2M}},
\end{aligned}
\end{equation}
where $L$ is the Lipschitz coefficient of the per-sample risk function $r$, $B_O$ is the upper bound of the spectral norm of the observable, $d$ is the data dimension, and $k$ is the number of qubits on which the encoding gate acts.

This generalization bound suggests that when using specific encoding strategies, the generalization error grows exponentially with data dimensionality. However, we will experimentally verify in Subsection~\ref{asubsec:gen_bound_encoding} that the generalization bound does not exhibit exponential growth, and our generalization bound is relatively tighter, while this bound consistently exceeds the maximum possible generalization error.

\subsection{Effects of Optimization Process on Generalization}
\label{subsec:gen_bound_optimization}

Furthermore, paper~\cite{yang2025stability} examines the relationship between generalization bounds and stability in quantum machine learning. According to Eq.(1) and Corollary 4.3 in~\cite{yang2025stability}, when the per-sample risk function is $L$-Lipschitz and its gradient is $v_L$-Lipschitz, the generalization bound for quantum machine learning optimized using SGD with a fixed learning rate $\eta$ satisfies:
\begin{equation}
  \label{eq:yang_bound}
\begin{aligned}
  \operatorname{gen}(h_S) \leqslant \frac{2\sqrt{2}L^2K B_O}{(LK B_O + \sqrt{2} v_L K B_O)M}(1+\eta(LKB_O + \sqrt{2}v_L K B_O))^{T},
\end{aligned}
\end{equation}
where $K$ is the number of parameterized gates in the quantum machine learning model, $B_O$ is the upper bound of the spectral norm of the observable, and $T$ is the number of training epochs.

This generalization bound suggests that as the number of training epochs increases, the generalization error grows exponentially with the number of training epochs. We experimentally refute this viewpoint in Subsection~\ref{subsec:gen_bound_epoch}, and this generalization bound is always larger than the maximum possible generalization error.

\resetAppendixCounters{D}

\section{Experiments Details}
\label{asec:experiments_details}
In this section, we provide the implementation details for the three figures presented in the main paper.

\subsection{Experiment on Quantum Phase Classification}
\label{subsec:gen_bound_classification}

For the experiments in Fig.~\ref{fig:gen_bound_classification}, we construct our dataset using ground states of a 6-qubit ANNNI model across varying parameter combinations $(\kappa,h)$. The dataset consists of quantum states $\rho_i$ paired with binary phase labels $y_i \in \{-1,1\}$, where $y_i = 1$ indicates ordered phase and $y_i = -1$ indicates disordered phase. To cover the entire phase diagram, we randomly select parameter pairs $(\kappa, h)$ uniformly from the plane shown in Fig.~\ref{fig:experiments}.(a). Each quantum state $\rho_i$ represents the numerically computed ground state of the corresponding ANNNI Hamiltonian. The quantum machine learning model employs the parameterized quantum circuit architecture illustrated in Figure~\ref{fig:experiments}.(b) with $L=20$ layers, where circuit parameters $\boldsymbol{\theta}$ are initialized from a standard Gaussian distribution. Predictions are obtained through measurement of the $Z_1$ observable on the first qubit. For training on dataset $S = \{(\rho^{(m)},y^{(m)})\}_{m=1}^{M}$, we employ the Hinge loss function:
\begin{equation}
  \label{eq:hinge_loss}
\begin{aligned}
\mathcal{L}(\boldsymbol{\theta};S) = \frac{1}{M} \sum_{m=1}^{M}  \max \left\{0, 1-y^{(m)} \operatorname{Tr}\left[ Z_{1} U(\boldsymbol{\theta}) \rho^{(m)} U(\boldsymbol{\theta})^{\dagger} \right] \right\}.
\end{aligned}
\end{equation}
We used the Adam optimizer with a learning rate of 0.005 and trained for 100 epochs. To compare the effect of different sample sizes on generalization capability, we selected training set sizes of $\{10,500,1000,1500,2000\}$ with a batch size of 200 (when the total training set size is smaller than the batch size, this becomes full batch training). Since the randomness originates from the training set sampling, for each sample size, we independently sampled the training set 10 times while maintaining identical initial circuit parameters. To ensure that the test error approximates the prediction error as closely as possible, we used a test set of 10,000 samples that were completely unseen during training, with the same test set being used across all training configurations (different training set sizes and different sampling iterations). When calculating the generalization error, we used the 0-1 risk as our per-sample risk function $r$, where the training and test accuracies are exactly 1 minus the training and test errors, respectively.

\subsection{Experiment on Comparison of Generalization Bounds}
\label{subsec:gen_bound_compare}

In the experiments shown in Fig.~\ref{fig:gen_bound_compare}, we systematically compare the generalization bound proposed in paper~\cite{caro2022generalization} with the generalization bound proposed in our paper. The experiments are divided into two parts: first, with fixed model complexity, we compare the impact of different training sample sizes on generalization capability; second, with fixed training sample size ($M=2000$), we examine the effect of different model complexities (achieved by adjusting the number of quantum circuit layers, where increasing layers corresponds to increasing the number of parameterized quantum gates) on generalization capability. The results are shown in Fig.~\ref{fig:gen_bound_compare}.(a) and (b), respectively.

Since the generalization bound in paper~\cite{caro2022generalization} is designed for a specific risk function (discussed in Subsection~\ref{subsec:gen_bound_classification}), we adopt the per-sample risk function defined in Eq.~\eqref{eq:C_loss_equivalent}, where $\rho_1 = \ket{0}\bra{0}$ and $\rho_{-1} = \ket{1}\bra{1}$, and $\boldsymbol{\theta}^{*}$ is the optimal parameters of the quantum machine learning model learned on the training set $S$. That is, in Fig.~\ref{fig:gen_bound_compare}, the generalization error is calculated using the risk function in Eq.~\eqref{eq:C_loss_equivalent}. We use the following loss function to learn the dataset $S = \{(\rho^{(m)}, y^{(m)})\}_{m=1}^{M}$:
$$
\mathcal{L}(\boldsymbol{\theta}; S) = \frac{1}{M} \sum_{m=1}^{M} \left(1 - \operatorname{Tr}\left[\rho_{y^{(m)}} U(\boldsymbol{\theta}) \rho^{(m)} U(\boldsymbol{\theta})^{\dagger}\right]\right).
$$
Both generalization bounds (Eq.~\eqref{eq:previous_bound} and Eq.~\eqref{eq:our_bound}) adopt a confidence level of $1 - \delta = 0.9$.

In Fig.~\ref{fig:gen_bound_compare}.(a), we directly use the prediction results from Fig.~\ref{fig:gen_bound_classification}, and change the risk function used for evaluating generalization error from the 0-1 risk function in Fig.~\ref{fig:gen_bound_classification} to the risk function in Eq.~\eqref{eq:C_loss_equivalent}. While in Fig.~\ref{fig:gen_bound_compare}.(b), to explore the influence of model parameters, we fix the training set and independently sample 10 different sets of initial parameters from a standard Gaussian distribution for experiments. In Fig.~\ref{fig:gen_bound_compare}.(b), the risk function used for evaluating generalization error is also Eq.~\eqref{eq:C_loss_equivalent}. The optimizer, learning rate, and batch size settings are all consistent with the experiments in Fig.~\ref{fig:gen_bound_classification}.

Furthermore, for classification tasks, the more commonly used risk function is the 0-1 risk, i.e., $r(h_S(\rho),y) = \mathbbm{1}(h_S(\rho) \neq y)$, which is 1 minus the 0-1 risk, namely the accuracy. It is worth noting that for the same problem, choosing different risk functions will result in different generalization bounds. For the 0-1 risk function, the generalization bound is $ 1 / \sqrt{M} + 3\sqrt{ \log (2 / \delta) / 2M}$, while for the risk function adopted in paper~\cite{caro2022generalization}, i.e., Eq.~\eqref{eq:C_loss_equivalent}, the generalization bound is $ 2 / \sqrt{M} + 3\sqrt{ \log (2 / \delta) / 2M}$. In Fig.~\ref{fig:gen_bound_layer}.(a)(b), we present the training and test results for the binary classification task from Fig.~\ref{fig:gen_bound_compare}.(b) using accuracy as the evaluation metric, along with the generalization bounds. Additionally, we supplement the training and test results using Eq.~\eqref{eq:C_loss_equivalent} as the evaluation metric in Fig.~\ref{fig:gen_bound_layer}.(c). For comparison, we show the generalization error and generalization bound when using Eq.~\eqref{eq:C_loss_equivalent} as the risk function in Fig.~\ref{fig:gen_bound_layer}.(d), compared to the results in Fig.~\ref{fig:gen_bound_layer}.(b).

Relative to the 0-1 risk, when using Eq.~\eqref{eq:C_loss_equivalent} as the risk function, the generalization bound is larger and the generalization error is also larger. This is mainly because the 0-1 risk only evaluates correctness, while Eq.~\eqref{eq:C_loss_equivalent} not only evaluates correctness but also measures the gap between predicted and true values. Additionally, subfigures (a) and (c) in Fig.~\ref{fig:gen_bound_layer} appear to exhibit a counterintuitive phenomenon: as the number of layers increases, both training and test accuracy improve, while the training and test errors measured using Eq.~\eqref{eq:C_loss_equivalent} also increase. This occurs because when the number of layers increases, on average, although the predictions deviate further from the true values, the number of correctly classified samples actually increases.

\begin{figure}[htpb]
  \centering
  \includegraphics[width=0.8\textwidth]{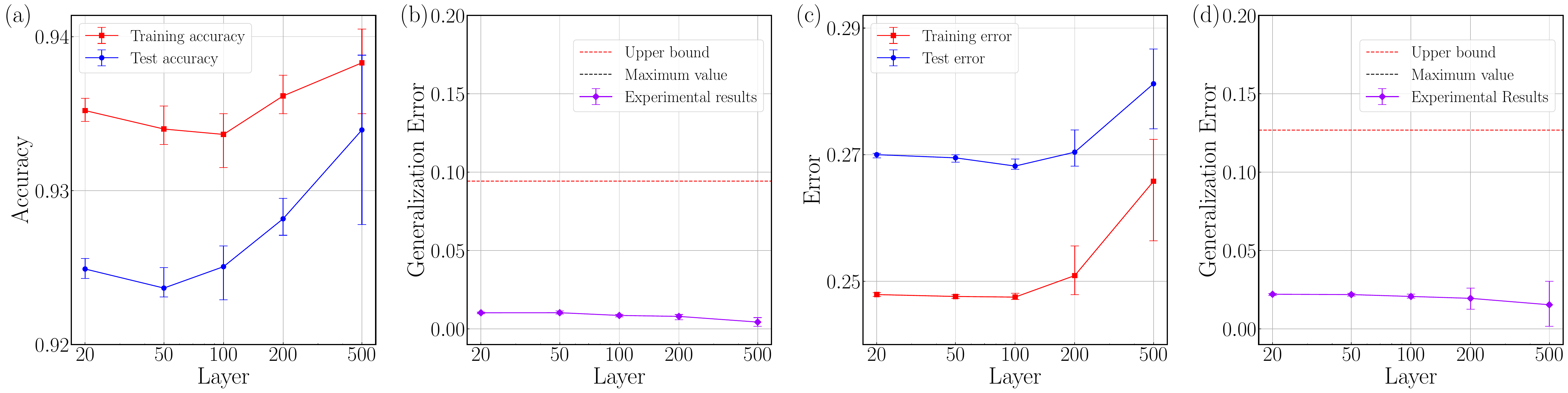}
  \caption{(a) Training accuracy and test accuracy in the experiment from Fig.~\ref{fig:gen_bound_compare}.(b). (b) Comparison between experimental generalization error and theoretical generalization bound with confidence level $1-\delta=0.9$, when using 0-1 risk function as the risk function and fixed sample size $M=2000$. (c) Training error and test error in the experiment from Fig.~\ref{fig:gen_bound_compare}.(b) measured using Eq.~\eqref{eq:C_loss_equivalent}. (d) Comparison between experimental generalization error and theoretical generalization bound with confidence level $1-\delta=0.9$, when using Eq.~\eqref{eq:C_loss_equivalent} as the risk function and fixed sample size $M=2000$. The error bars represent the minimum and maximum values across 10 independent runs with different random seeds, with the central line showing the mean value.}
  \label{fig:gen_bound_layer}
\end{figure}

\subsection{Experiment on Random Label}

In Fig.~\ref{fig:random_label}, when creating datasets with random labels, for different sample sizes, we first sampled 10 training sets, then randomly assigned labels $\{-1,1\}$ to each quantum state in the dataset according to a uniform distribution for each training set. Each experiment used the same test set, and all other settings remained identical to those in the experiment described in Fig.~\ref{fig:gen_bound_classification}. Here, we used the 0-1 risk as the per-sample risk function.

\resetAppendixCounters{E}

\section{Generalization in Regression Task}
\label{asec:regression_experiments}

This section examines the application of our proposed generalization upper bound to quantum machine learning regression models. We focus on evaluating the effects of data dimensionality (the number of qubits) and data encoding methods on generalization capability, and compare with the generalization bound for the special encoding method from paper~\cite{caro2021encodingdependent} discussed in Subsection~\ref{asubsec:gen_bound_encoding}.

\subsection{Regression Experiments}
\label{asubsec:regression_experiments}

Furthermore, we conducted experiments on regression problems. We chose the target function for regression as $f(\boldsymbol{x}) = 1 - \boldsymbol{x}^{\top} \boldsymbol{x} /d$, where $\boldsymbol{x} \in \mathbb{R}^{d}$ and $d$ is the vector dimension. When each dimension of the data $\boldsymbol{x}$ is uniformly sampled from $[-1,1]$, we have $f(\boldsymbol{x}) \in [0,1]$. We used angle encoding through quantum gates $R_y(x) = e^{-i x Y}$, where $Y$ is the Pauli $Y$ matrix, with $N$ qubits encoding $N$-dimensional data using the circuit architecture shown in Fig.~\ref{fig:regression_circuit}.(a). Let $\rho(\boldsymbol{x}_i)$ represent the quantum state corresponding to data $\boldsymbol{x}_i$ after angle encoding, and $y_i$ be the label $f(\boldsymbol{x}_i)$ corresponding to data $\boldsymbol{x}_i$. For the $N$-dimensional function regression task, we used an $N$-qubit circuit with the observable $O_Z = Z_{1} \otimes \cdots \otimes Z_N$. In this scenario, our proposed generalization bound is shown in Eq.~\eqref{eq:regression_bound}.

\begin{figure}[htpb]
  \centering
  \includegraphics[width=0.75\textwidth]{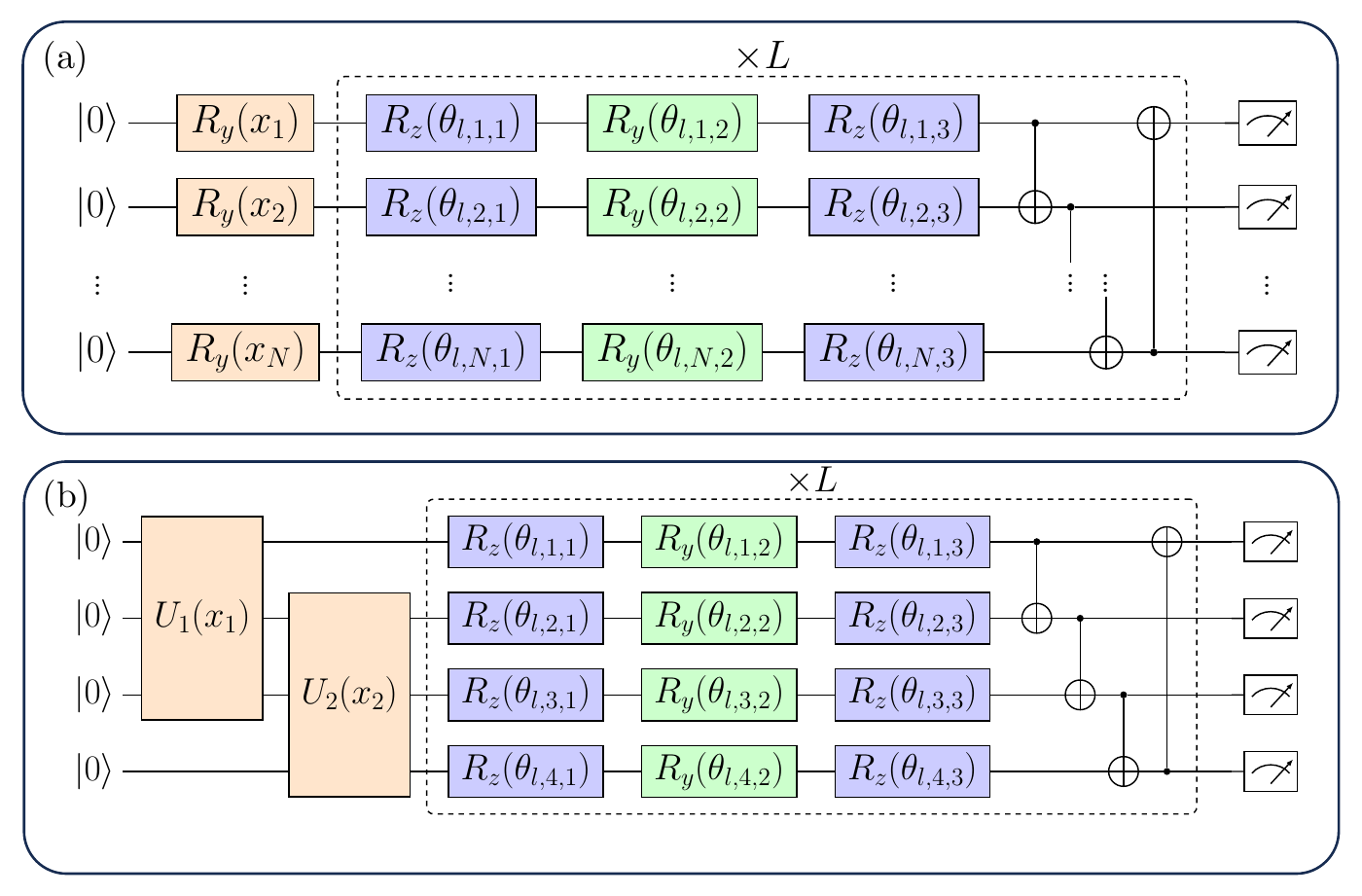}
  \caption{(a) The quantum circuit first encodes $N$-dimensional data via $R_y(x_i)$ angle encoding, then applies an $L$-layer parameterized circuit. Each layer consists of $R_z(\theta_1) R_y(\theta_2) R_z(\theta_3)$ rotations and ring-topology CNOT gates for entanglement. (b) Same variational architecture as (a) but with the special encoding method, where each encoding gate $U_i(x_i) = e^{-i x_i H_i}$ has a different Hamiltonian $H_i$, with $H_i = \operatorname{diag}\left( (i+2),2(i+2), \cdots,2^3(i+2) \right)$ containing $2^3$ distinct eigenvalues.}
  \label{fig:regression_circuit}
\end{figure}

For training on the dataset $S = \{(\rho(\boldsymbol{x}^{(m)}),y^{(m)})\}_{m=1}^{M}$, we used the mean squared error (MSE) loss function:
$$
\begin{aligned}
\mathcal{L}(\boldsymbol{\theta};S) = \frac{1}{M} \sum_{m=1}^{M} \left(y^{(m)} - \operatorname{Tr}\left[O_Z U(\boldsymbol{\theta}) \rho(\boldsymbol{x}^{(m)}) U(\boldsymbol{\theta})^{\dagger} \right] \right)^2.
\end{aligned}
$$
When evaluating the regression error, we used the mean absolute error as our per-sample risk function, i.e., $r(h_S(x),y) = |y-h_S(x)|$. All other experimental settings were the same as in the experiment described in Fig.~\ref{fig:gen_bound_classification}. In this experiment, we used 6 qubits to regress a 6-dimensional function with $L = 20$, fixed the initial parameter distribution, and randomly sampled 10 different datasets. The experimental results for the regression task are shown in Fig.~\ref{fig:regression_results}.

\begin{figure}[htpb]
  \centering
  \includegraphics[width=0.6\textwidth]{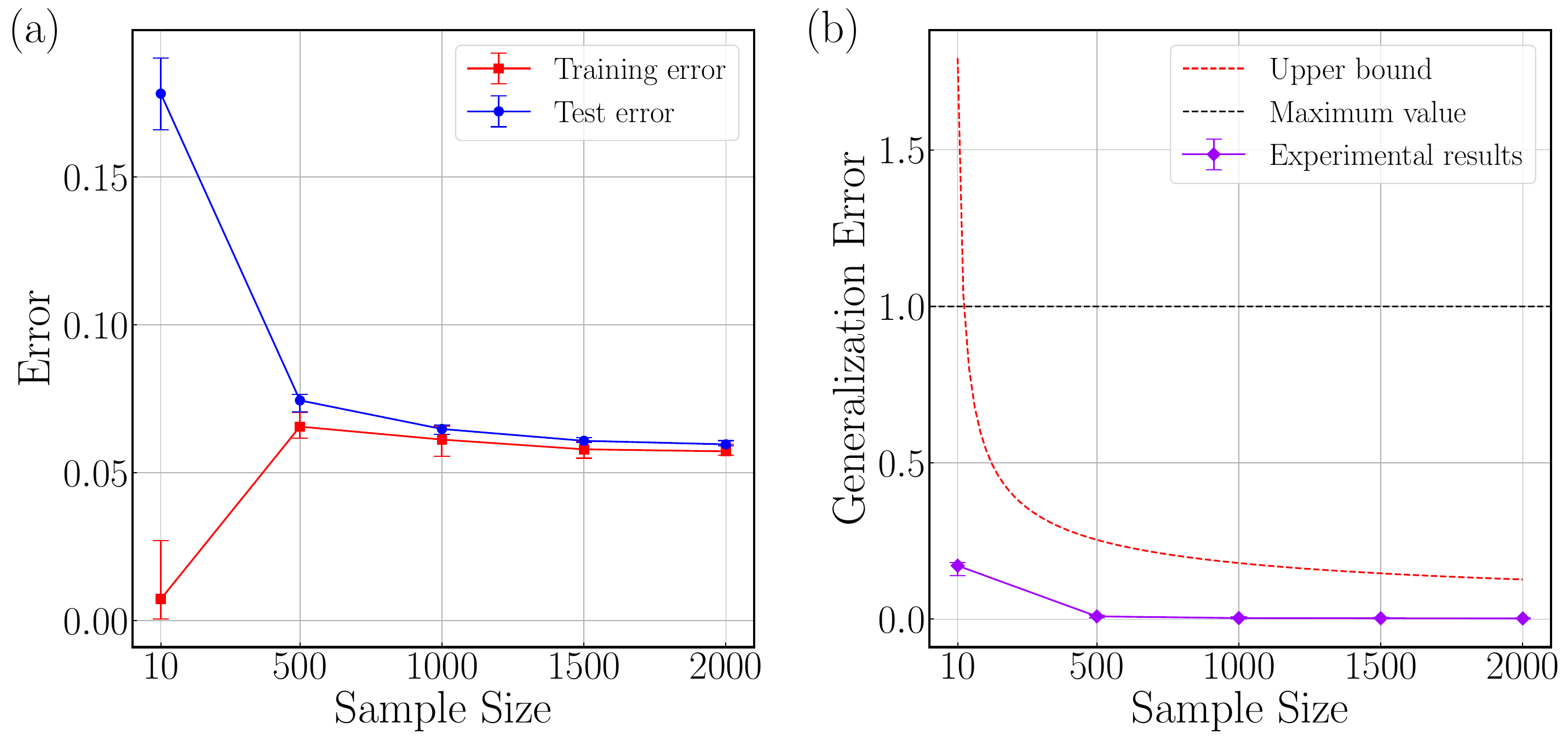}
  \caption{(a) Training error and test error under different sample sizes for regression tasks; (b) Comparison between experimental generalization error and theoretical generalization upper bound with confidence level $1-\delta=0.9$. The error bars represent the minimum and maximum values across 10 independent runs with different training sets, with the central line showing the mean value.}
  \label{fig:regression_results}
\end{figure}

\subsection{Effects of Data Dimension (or Qubit Number) on Generalization}
\begin{figure}[htpb]
  \centering
  \includegraphics[width=0.75\textwidth]{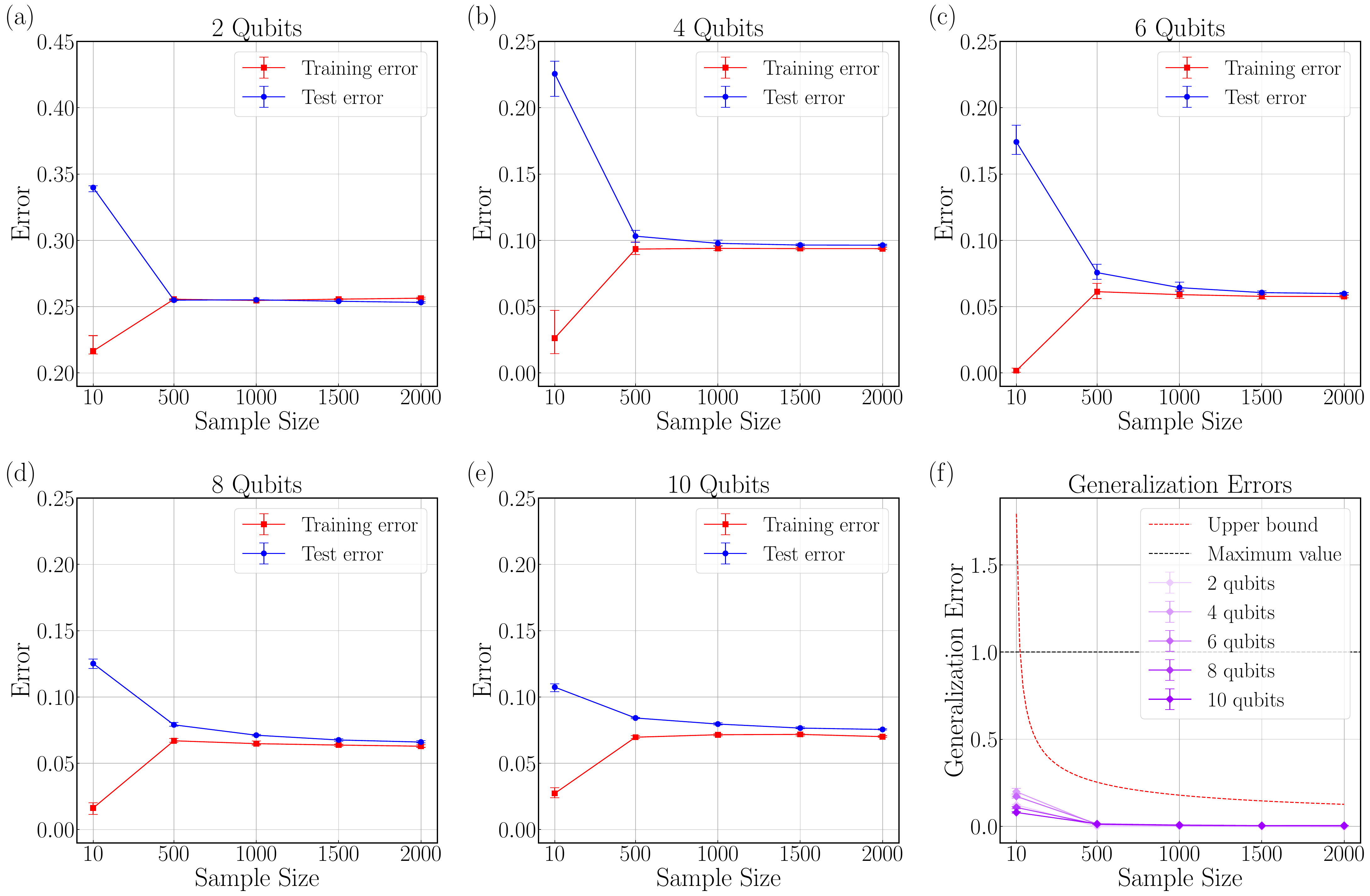}
  \caption{(a)-(e) Training error and test error for regression tasks under different numbers of qubits (different data dimensions). (f) Comparison between experimental generalization error and our proposed generalization upper bound with confidence level $1-\delta=0.9$ for different numbers of qubits. The error bars represent the minimum and maximum values across 10 independent runs with different random seeds, with the central line showing the mean value.}
  \label{fig:regression_qubit}
\end{figure}
To investigate the impact of data dimensionality on generalization ability in quantum machine learning, we conducted regression experiments using quantum circuits with varying numbers of qubits. Since the regression circuit architecture in Fig.~\ref{fig:regression_circuit} requires one qubit per data dimension, we tested systems with $2, 4, 6, 8, 10$ qubits to regress functions of corresponding dimensions $d = 2, 4, 6, 8, 10$. For each configuration, we fixed the dataset and performed 10 independent runs with different randomly initialized parameters sampled from a standard Gaussian distribution. All other experimental settings remained consistent with those described in Fig.~\ref{fig:regression_results}.

The results are shown in Fig.~\ref{fig:regression_qubit}. It can be observed that as the number of qubits increases, or equivalently as the data dimensionality increases, the experimental generalization error of the quantum machine learning model remains at a stable level, and even shows a decreasing trend when the number of qubits is larger.

\begin{figure}[htpb]
  \centering
  \includegraphics[width=0.75\textwidth]{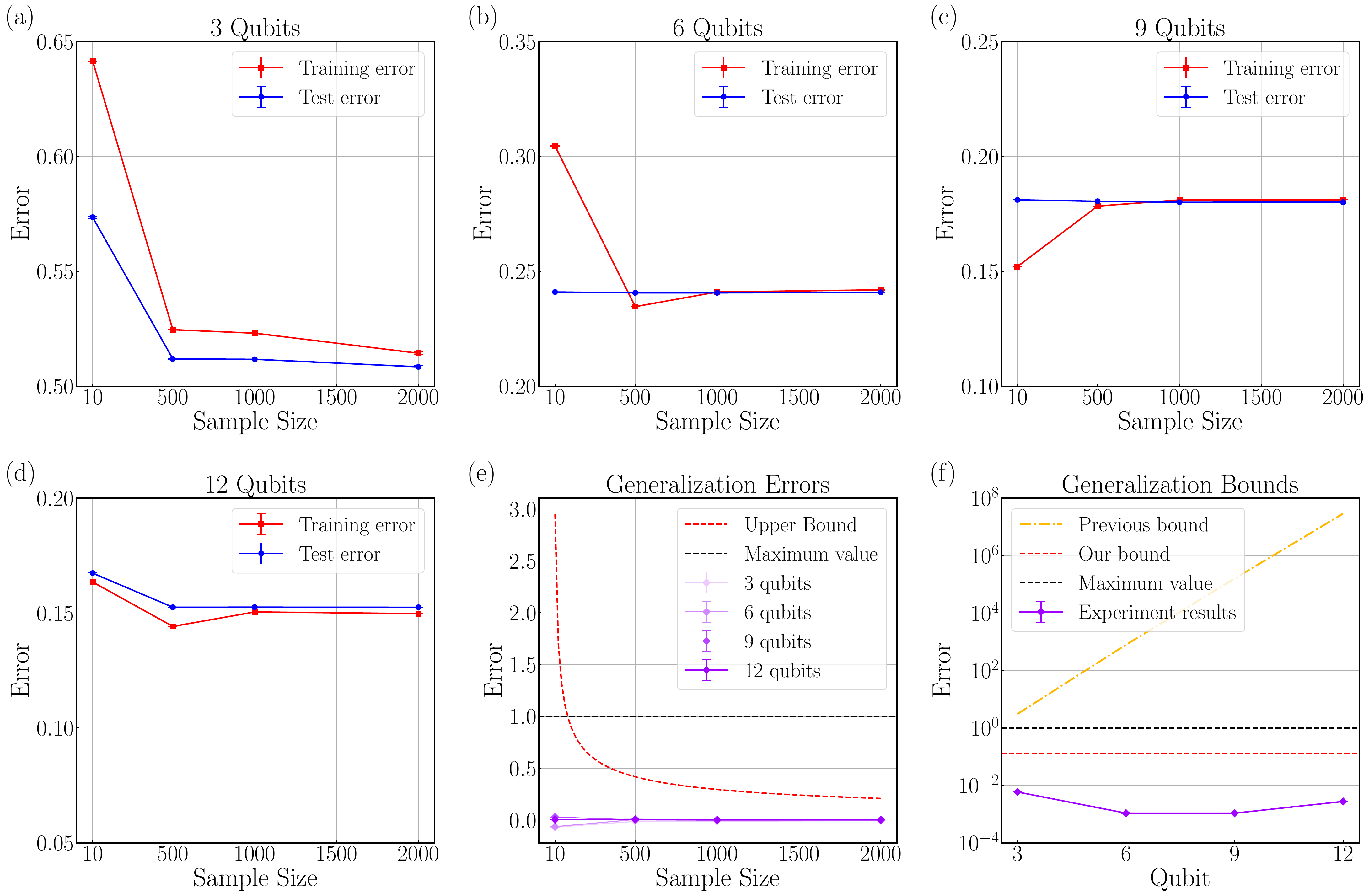}
  \caption{(a)-(d) Training error and test error for regression tasks under different numbers of qubits for special encoding method described in Fig.~\ref{fig:regression_circuit}.(b). (e) Comparison between experimental generalization error and our theoretical generalization upper bound with confidence level $1-\delta=0.9$ for different data dimensions. (f) Comparison between our theoretical upper bound and the bound from paper~\cite{caro2021encodingdependent} as data dimensionality varies with fixed sample size. Both bounds are shown with confidence $1-\delta = 0.9$. The upper bound proposed in~\cite{caro2021encodingdependent} is always larger than the maximum possible generalization error. The error bars represent the minimum and maximum values across 10 independent runs with different random seeds, with the central line showing the mean value.}
  \label{fig:regression_special_qubit}
\end{figure}

\subsection{Effects of Data Encoding on Generalization}
\label{asubsec:gen_bound_encoding}

Furthermore, paper~\cite{caro2021encodingdependent} indicates that when using encoding of the form $U = e^{\mathrm{i} x H}$, if $H$ has an exponential number of distinct eigenvalues, the generalization bound also increases exponentially with data dimensionality. Here, we used 3-qubit encoding gates, applying different encoding gates $U_i(x_i) = e^{\mathrm{i} x_i H_i}$ for each data point $x_i$, where $H_i = \operatorname{diag}\left( (i+2),2(i+2), 3(i+2), \cdots,2^3(i+2) \right)$, containing $2^3$ distinct eigenvalues. The circuit is shown in Fig.~\ref{fig:regression_circuit}.(b), demonstrating that $N$ qubits can encode $d=N-2$ dimensional data. We encoded regression data of dimensions $1,4,7,10$ using $3,6,9,12$ qubits respectively. For each configuration, we fixed the dataset and performed 10 independent runs with different randomly initialized parameters sampled from a standard Gaussian distribution. All other experimental settings remained consistent with those described in Fig.~\ref{fig:regression_results}. Finally, we fixed the sample size at $2000$ and examined whether the generalization error would increase exponentially with data dimensionality by checking the generalization errors for different numbers of qubits (different data dimensions).

As shown in Fig.~\ref{fig:regression_special_qubit}, the experimental generalization error does not increase exponentially, remains below our upper bound, and aligns more closely with our upper bound. Furthermore, in regression problems, since $f(\boldsymbol{x}) \in [0,1]$, the maximum generalization error is $1$, while the generalization bound proposed in~\cite{caro2021encodingdependent}, namely Eq.~\eqref{eq:gen_bound_encoding}, is always larger than the maximum possible generalization error. It is worth noting that in the numerical calculation of the generalization bound Eq.~\eqref{eq:gen_bound_encoding}, we omitted the $\int_{0}^{\gamma_{0}} \sqrt{\log \left(2 / \beta\right)} \mathrm{d} \beta$ term for computational convenience.

\resetAppendixCounters{F}

\section{Effects of Optimization Process on Generalization}
\label{asec:optimization}
Additionally, we experimentally verified the impact of batch Size, epochs, learning rate, and optimizers on generalization capability during the optimization process.

\subsection{Effects of Batch Size on Generalization}

When studying the effect of batch size on generalization capability, we used the same dataset in the main paper. We selected a training set sample size of $M=2000$ and batch sizes of $\{1, 200, 500, 1000, 2000\}$. In this experiment, we fixed the dataset and ran 10 times with different initial parameters sampled from a standard Gaussian distribution. Other experimental settings remained consistent with those described in Fig.~\ref{fig:gen_bound_classification}. The experimental results are shown in the Fig.~\ref{fig:batch_size}.
\begin{figure}[htpb]
  \centering
  \includegraphics[width=0.63\textwidth]{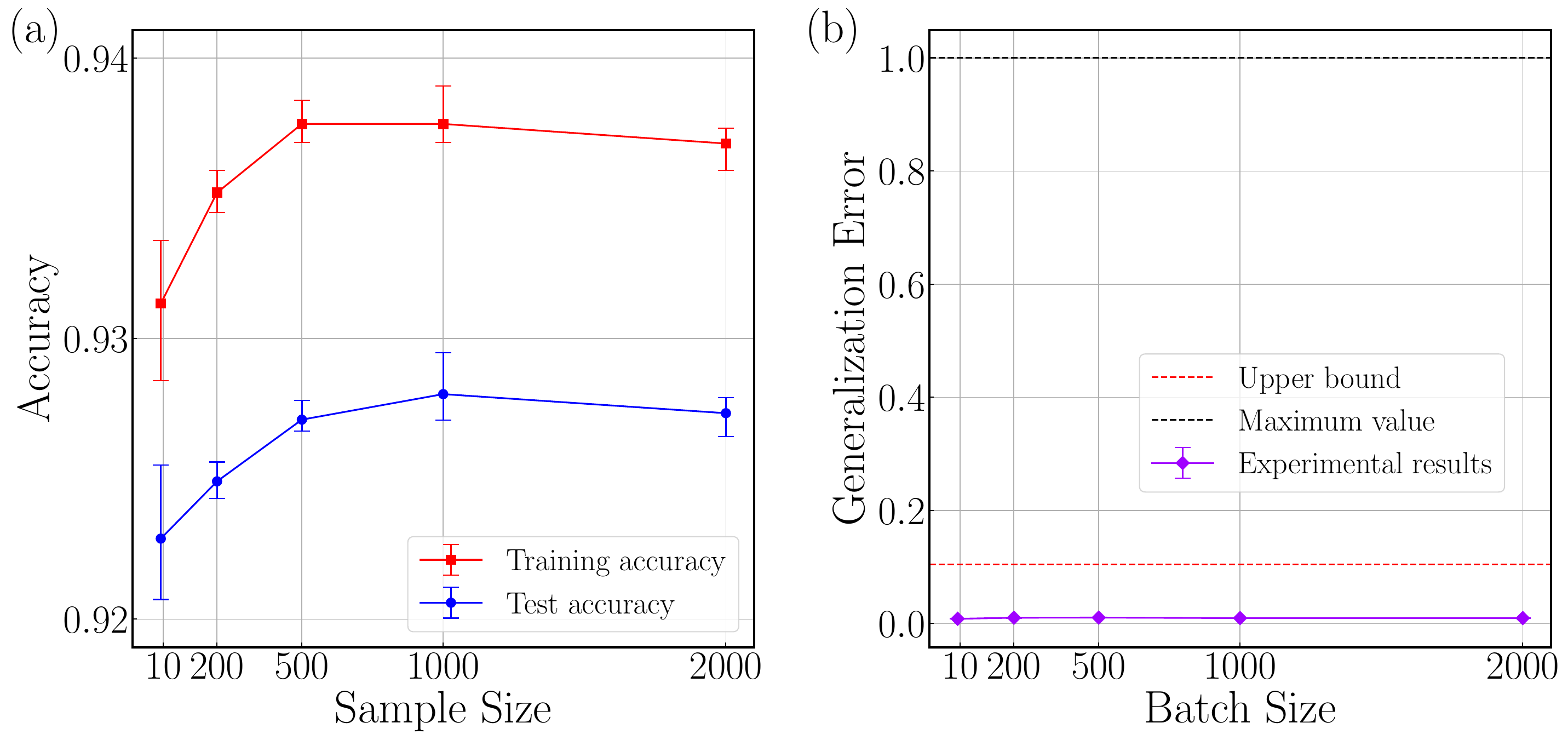}
  \caption{(a) Training accuracy and test accuracy under different batch sizes; (b) Comparison between experimental generalization error and theoretical generalization upper bound with confidence level $1-\delta=0.9$ and fixed sample size $M=2000$. The error bars represent the minimum and maximum values across 10 independent runs with different random seeds, with the central line showing the mean value.}
  \label{fig:batch_size}
\end{figure}

\subsection{Effects of Epoch on Generalization}
\label{subsec:gen_bound_epoch}

\begin{figure}[htpb]
  \centering
  \includegraphics[width=0.63\textwidth]{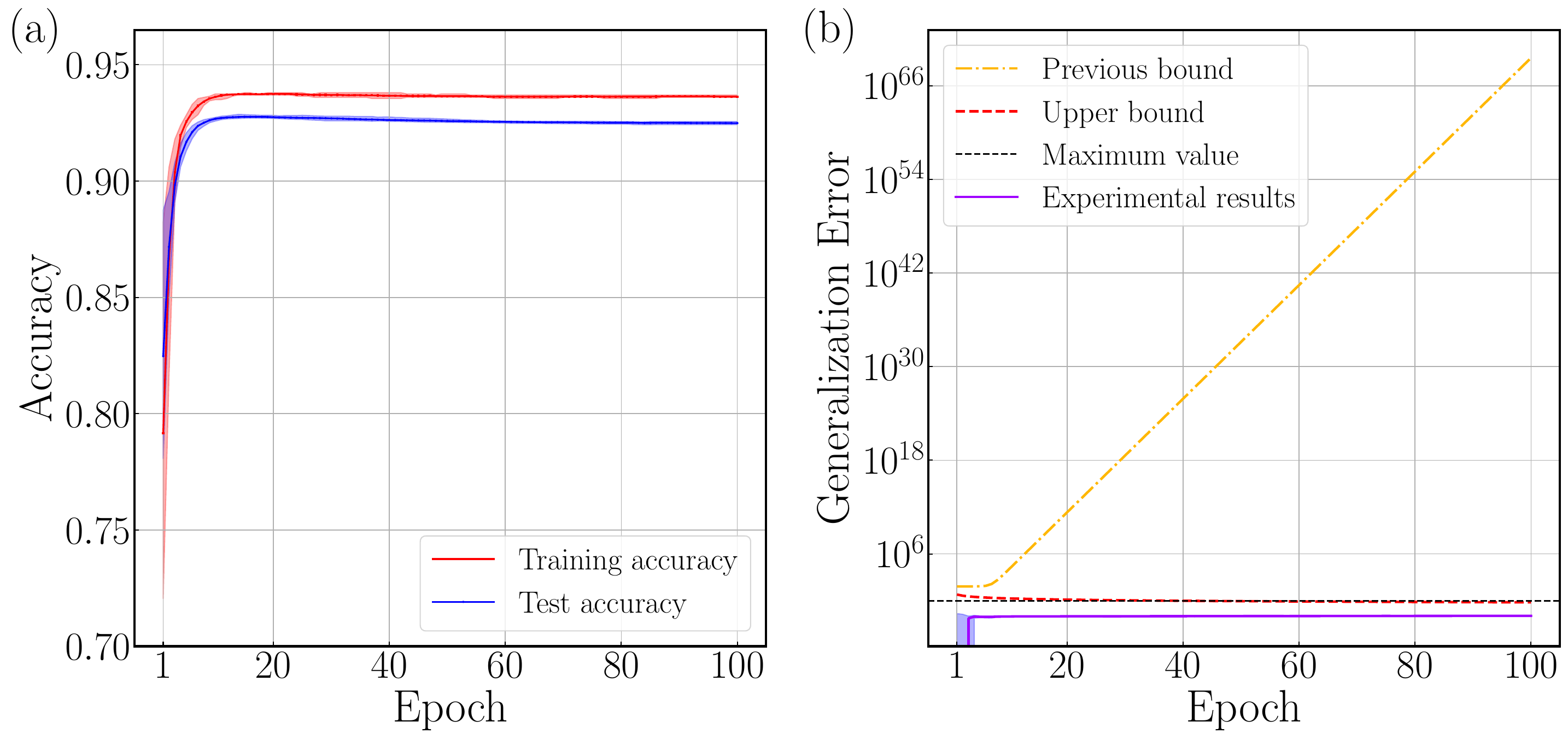}
  \caption{(a) Training accuracy and test accuracy under different number of epochs; (b) Comparison between experimental generalization error and theoretical generalization upper bound  with confidence level $1-\delta=0.9$ and fixed sample size $M=2000$. The shaded area represents the minimum and maximum values across 10 independent runs with different random seeds, with the central line showing the mean value.}
  \label{fig:epoch_bound}
\end{figure}
Furthermore, we experimentally investigated the impact of the number of training epochs on generalization capability and compare our proposed generalization upper bound with the bound from~\cite{yang2025stability} discussed in the Subsection~\ref{subsec:gen_bound_optimization}. To be consistent with the theoretical assumptions of bound in Eq.~\eqref{eq:yang_bound}, we chose SGD as the optimizer with a learning rate of 0.005, fixed the dataset with a sample size of $2000$, and ran 10 times with different initial parameters sampled from a standard Gaussian distribution. We displayed the training accuracy, test accuracy, and generalization error for each epoch across the 10 experiments and compared them with the theoretical results from~\cite{yang2025stability}. The experimental results are shown in the Fig.~\ref{fig:epoch_bound}:

Since the generalization bound in Eq.~\eqref{eq:yang_bound} grows exponentially with the number of epochs, this bound is much larger than the actual generalization error, and even larger than the maximum possible generalization error.

\subsection{Effects of Learning Rate on Generalization}

We investigated the impact of learning rate on experimental generalization capability. In this experiment, we fixed the dataset and ran 10 times with different initial parameters sampled from a standard Gaussian distribution. Learning rates were sampled from $\{0.0005, 0.005, 0.05, 0.5, 5\}$, with all other settings identical to those in the experiment described in Fig.~\ref{fig:gen_bound_classification}. The experimental results are shown in the Fig.~\ref{fig:lr_bound}:

\begin{figure}[htpb]
  \centering
  \includegraphics[width=0.8\textwidth]{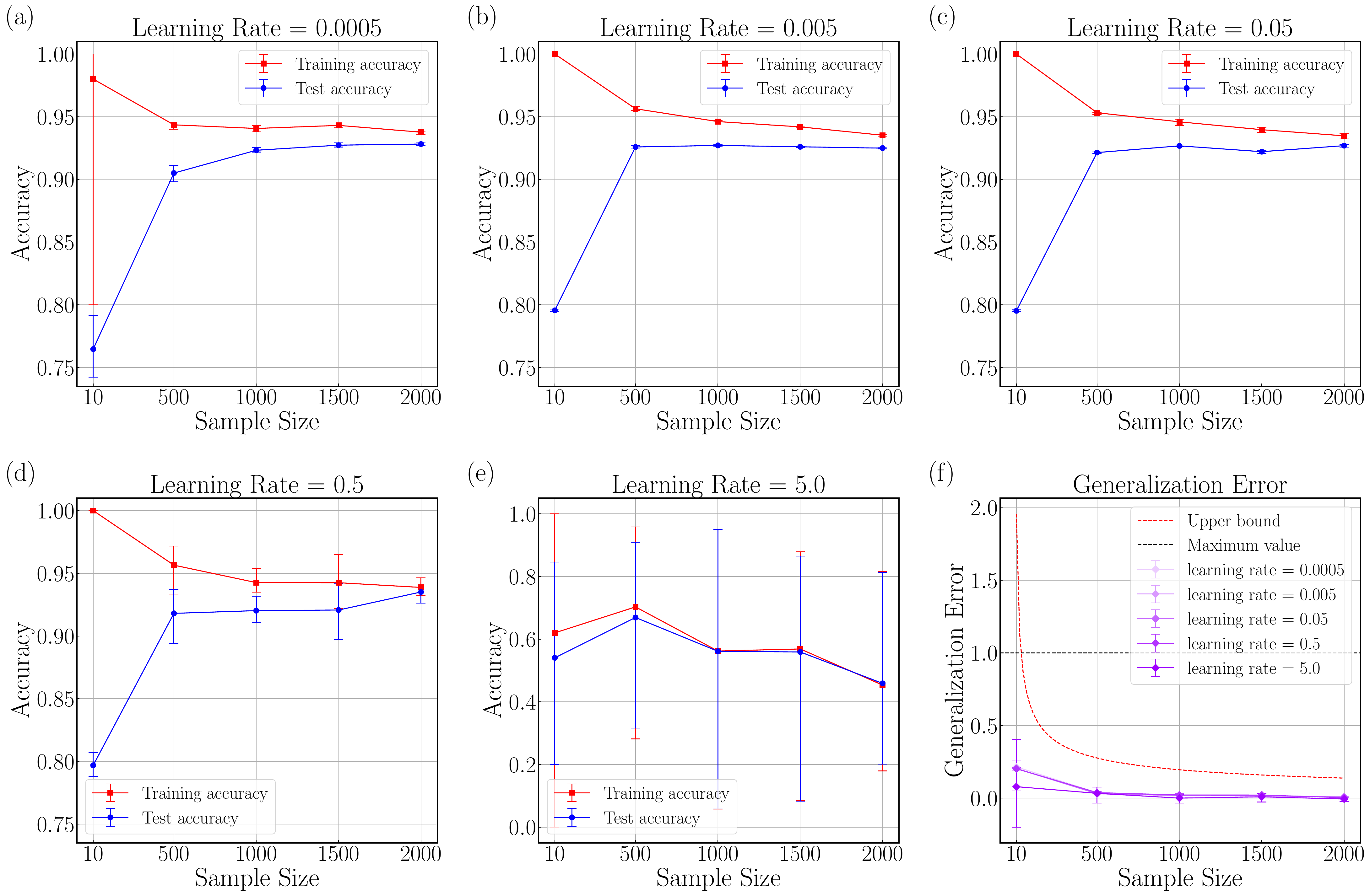}
  \caption{(a)-(e) Training accuracy and test accuracy under different learning rates; (f) Comparison between experimental generalization error and theoretical generalization upper bound with confidence level $1-\delta=0.9$. The error bars represent the minimum and maximum values across 10 independent runs with different random seeds, with the central line showing the mean value.}
  \label{fig:lr_bound}
\end{figure}

\subsection{Effects of Optimizer on Generalization}

We investigated the impact of commonly used optimizers on generalization capability. In this experiment, we fixed the dataset and ran 10 times with different initial parameters sampled from a standard Gaussian distribution. Optimizers were selected from \{\text{SGD}, \text{Adam}, \text{RMSprop}, \text{AdaGrad}, \text{Lion}\}, with all other settings identical to those in the experiment described in Fig.~\ref{fig:gen_bound_classification}. The experimental results are shown in the Fig.~\ref{fig:optimizer_bound}:

\begin{figure}[htpb]
  \centering
  \includegraphics[width=0.8\textwidth]{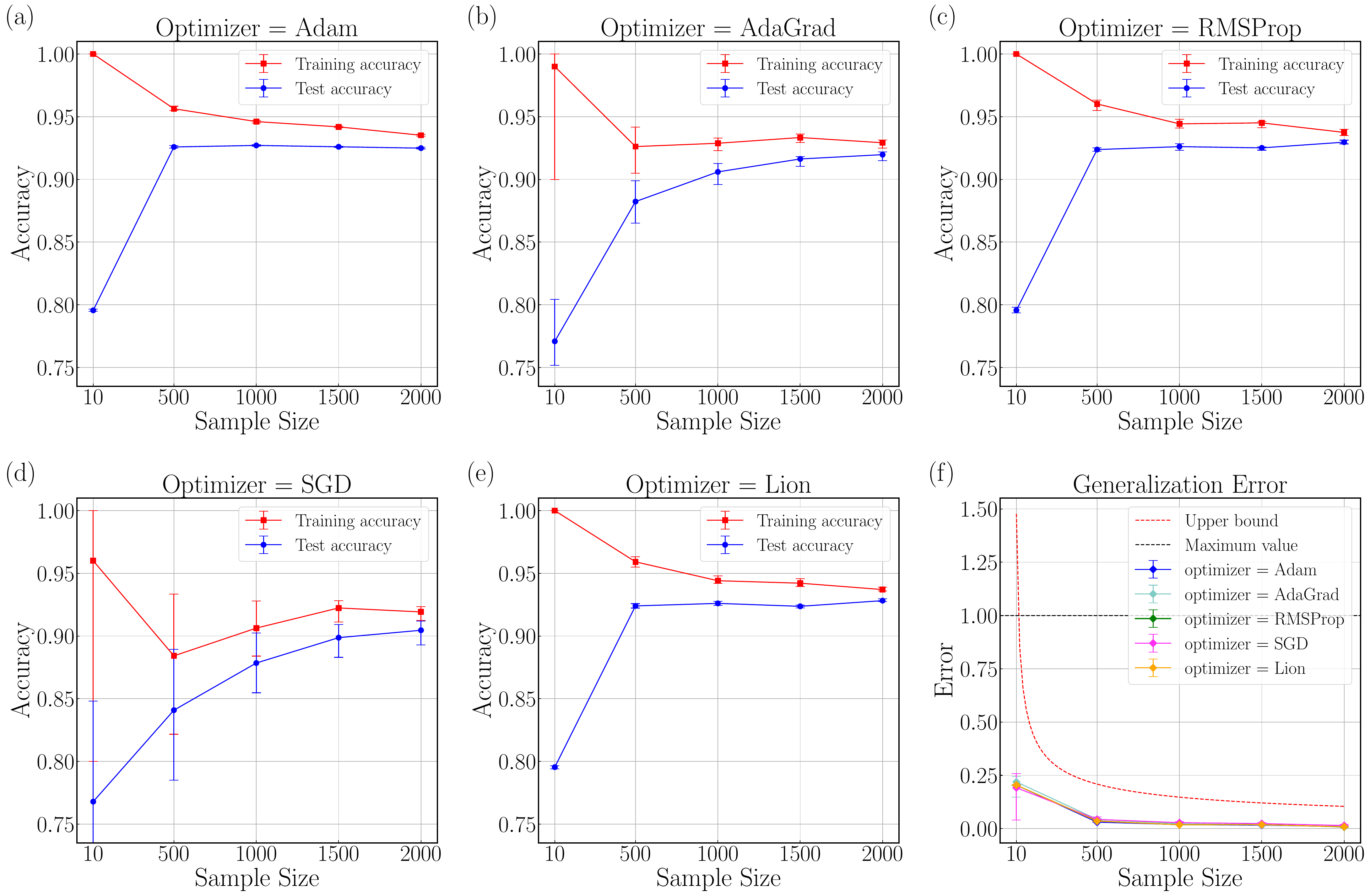}
  \caption{(a)-(e) Training accuracy and test accuracy under different optimizers; (f) Comparison between experimental generalization error and theoretical generalization upper bound with confidence level $1-\delta=0.9$. The error bars represent the minimum and maximum values across 10 independent runs with different random seeds, with the central line showing the mean value.}
  \label{fig:optimizer_bound}
\end{figure}

\end{document}